\documentclass{article}

\usepackage[english]{babel}
\usepackage[T1]{fontenc}

\usepackage[letterpaper,top=3.5cm,bottom=3.5cm,left=2.5cm,right=2.5cm,marginparwidth=2cm]{geometry}

\usepackage{amsmath}
\usepackage{graphicx}
\usepackage{amsfonts}
\usepackage{amsthm}
\usepackage{comment}
\usepackage{mathrsfs}
\usepackage[colorinlistoftodos]{todonotes}
\usepackage[colorlinks=true, allcolors=blue]{hyperref}
\usepackage{amssymb}
\newtheorem{lemma}{Lemma}[section]
\newtheorem{theorem}{Theorem}[section]
\newtheorem{definition}{Definition}[section]
\newtheorem{remark}{Remark}[section]
\newtheorem{corollary}{Corollary}[section]
\newtheorem{proposition}{Proposition}[section]
\numberwithin{equation}{section}

\title{Scattering for the Wave Equation on de Sitter Space \\ in All Even Spatial Dimensions}
\author{Serban Cicortas\footnote{Princeton University, Department of Mathematics, Fine Hall, Washington Road, Princeton, NJ 08544, USA}}

\begin{document}

\maketitle

\begin{abstract}
    For any $n\geq4$ even, we establish a complete scattering theory for the linear wave equation on the $(n+1)$-dimensional de Sitter space. We prove the existence and uniqueness of scattering states, and asymptotic completeness. Moreover, we construct the scattering map taking asymptotic data at past infinity $\mathcal{I}^-$ to asymptotic data at future infinity $\mathcal{I}^+$. Identifying $\mathcal{I}^-$ and $\mathcal{I}^+$ with $S^n,$ we prove that the scattering map is a Banach space isomorphism on $H^{s+n}(S^n)\times H^{s}(S^n),$ for any $s\geq1.$

    The main analysis is carried out at the level of the model equation obtained by differentiating the linear wave equation $\frac{n}{2}$ times in the time variable. The main result of the paper follows from proving a scattering theory for this equation. In particular, for the model equation we construct a scattering isomorphism from asymptotic data in $H^{s+\frac{1}{2}}(S^n)\times H^s(S^n)\times H^s(S^n)$ to Cauchy initial data in $H^{s+\frac{1}{2}}(S^n)\times H^{s+\frac{1}{2}}(S^n)\times H^{s-\frac{1}{2}}(S^n)$.
\end{abstract}

\textbf{Keywords} Wave Equation, de Sitter Space, Scattering Theory, Hyperbolic PDEs on Manifolds.

\textbf{Mathematics Subject Classification} 58J45, 35L05, 35P25, 35Q75.

\section{Introduction}

For any $n\geq3,$ we consider the $(n+1)$-dimensional de Sitter space $\big(\mathbb{R}\times S^n,g_{dS}\big)$ with metric:
\begin{equation}\label{de Sitter metric}
    g_{dS}=-dT^2+\cosh^2Td\sigma_n^2
\end{equation}
where $d\sigma_n^2$ denotes the standard metric on $S^n.$ We denote the past infinity $\{T\rightarrow-\infty\}$ by $\mathcal{I}^-,$ and future infinity $\{T\rightarrow\infty\}$ by $\mathcal{I}^+.$  Both $\mathcal{I}^-$ and $\mathcal{I}^+$ can be identified with $S^n.$

The metric (\ref{de Sitter metric}) introduced in \cite{deSitter} is the ground state solution to the Einstein vacuum equations with positive cosmological constant $\Lambda=\frac{n(n-1)}{2}:$
\begin{equation}\label{EVE}
    Ric_{\mu\nu}-\frac{1}{2}Rg_{\mu\nu}+\Lambda g_{\mu\nu}=0
\end{equation}

The study of de Sitter space has been of great interest in mathematical general relativity, starting with the works of Friedrich \cite{Friedrich1}-\cite{Friedrich2}, who proved that the $(3+1)$-dimensional de Sitter space is non-linearly stable to small perturbations of the asymptotic data at $\mathcal{I}^-.$ The argument uses the essential fact that in $(3+1)$ dimensions de Sitter space has a smooth conformal compactification. This method was generalized by Anderson \cite{Anderson} to all odd spatial dimensions, once again making use of the smooth conformal compactification. Both works establish a scattering theory in a neighborhood of de Sitter space by proving:

\begin{enumerate}
    \item \textit{Existence and uniqueness of scattering states}: for any suitable asymptotic data at $\mathcal{I}^-$ in a neighborhood of the de Sitter data at $\mathcal{I}^-$, there exists a unique solution of (\ref{EVE}) which is globally close to (\ref{de Sitter metric}).
    \item \textit{Asymptotic completeness}: any regular solution of (\ref{EVE}) which is sufficiently close to de Sitter space induces unique asymptotic data at $\mathcal{I}^-$ and $\mathcal{I}^+$.
    \item \textit{Boundedness of the scattering map}: for the above solution, the map taking the asymptotic data at $\mathcal{I}^-$ to the asymptotic data at $\mathcal{I}^+$ is bounded.
\end{enumerate}
We refer the reader to \cite{DRSR} and \cite{ReedSimon} for a more detailed introduction to scattering theory. 
 
The above argument does not apply in even spatial dimensions or in the presence of general types of matter, because of the lack of smoothness of the conformal compactification. A robust proof of stability in all dimensions in the more general setting of the Einstein equations coupled to a non-linear scalar field was given by  Ringström in \cite{Ringstrom1}. Unlike the above works, \cite{Ringstrom1} only considers Cauchy data on a finite time slice and does not address the problem of scattering.

Our discussion highlights the fact that without the conformal method it is significantly more challenging to prove a scattering theory for de Sitter space. Some of the difficulties that one faces in the even dimensional case are already present for the linear wave equation on a fixed de Sitter background with $n\geq4$ even:
\begin{equation}\label{wave equation de sitter intro}
    \square_{g_{dS}}\Tilde{\phi}=0
\end{equation}

The purpose of this paper is to prove a complete scattering theory for (\ref{wave equation de sitter intro}).

\subsection{Main Result}

In this section, we state the main theorem of the paper. Beforehand, we briefly introduce the formal analysis of (\ref{wave equation de sitter intro}), which provides important insights into the statement of the theorem. 

The formal expansion of $\Tilde{\phi}$ at $\mathcal{I}^-$ is given by:
\begin{equation}\label{expansion for phi intro}
    \Tilde{\phi}(T)=\sum_{k=0}^{n/2-1}\frac{1}{k!}\phi_k\cdot e^{2kT}+\frac{1}{(n/2)!}\mathcal{O}\cdot Te^{nT}+\frac{1}{(n/2)!}\Tilde{h}\cdot e^{nT}+O\big(T^2e^{(n+2)T}\big)
\end{equation}

This is an expansion in $e^T$ and $T,$ which illustrates the lack of smoothness at infinity. The coefficient $\mathcal{O}$ is the analogue of the Fefferman-Graham obstruction tensor \cite{Fefferman-Graham}, and is an essential feature of the even spatial dimension case. On the contrary, in the odd dimensional case the above expansion would only contain powers of $e^T$, allowing for a smooth compactification.

By further studying formal solutions of (\ref{wave equation de sitter intro}) of the form (\ref{expansion for phi intro}), we obtain a series of compatibility relations. These imply that each $\phi_a$ is determined by $\phi_0$ and its derivatives, and to leading order term we have that $\phi_a\sim\Delta^{a}\phi_0+\ldots,$ for all $1\leq a\leq\frac{n}{2}-1$. Here $\Delta$ represents the Laplacian of the standard metric on $S^n.$ Similarly, we also have that $\mathcal{O}\sim\Delta^{\frac{n}{2}}\phi_0+\ldots$ is determined by $\phi_0.$ Moreover, we get that each $O\big(T^2e^{(n+2)T}\big)$ term in (\ref{expansion for phi intro}) is formally determined by $\phi_0$ and $\Tilde{h}$. This suggests that we could consider $\big(\phi_0,\Tilde{h}\big)$ to represent asymptotic data at $\mathcal{I}^-.$ While this choice would parameterize the space of smooth asymptotic data, it does not capture the quantitative properties of the solution. A key aspect that will be clear from our analysis is the need to renormalize $\Tilde{h}$ by setting $\mathfrak{h}=\Tilde{h}-(\log\nabla)\mathcal{O},$ where $(\log\nabla)\mathcal{O}$ is defined using a suitable Fourier multiplier. We refer the reader to Section \ref{scattering in de sitter section} for the precise definition of asymptotic data.

We state the main result of the paper:

\begin{theorem}[Scattering Theory for the Wave Equation]\label{main theorem intro}
    For any $n\geq4$ even integer, we have a complete scattering theory for (\ref{wave equation de sitter intro}):

    \begin{enumerate}
    \item Existence and uniqueness of scattering states: for any $\phi_0,\mathfrak{h}\in C^{\infty}(S^n)$, there exists a unique smooth solution $\Tilde{\phi}:\mathbb{R}\times S^n\rightarrow\mathbb{R}$ of (\ref{wave equation de sitter intro}) with asymptotic data at $\mathcal{I}^-$ given by $\big(\phi_0,\mathfrak{h}\big)$.
    \item Asymptotic completeness: any smooth solution $\Tilde{\phi}:\mathbb{R}\times S^n\rightarrow\mathbb{R}$ of (\ref{wave equation de sitter intro}) induces $\phi_0,\mathfrak{h}\in C^{\infty}(S^n)$ unique asymptotic data at $\mathcal{I}^-$ and $\underline{\phi_0},\underline{\mathfrak{h}}\in C^{\infty}(S^n)$ unique asymptotic data at $\mathcal{I}^+.$
    \item The scattering isomorphism: for the above solution, we define the scattering map $\big(\phi_0,\mathfrak{h}\big)\mapsto\big(\underline{\phi_0},\underline{\mathfrak{h}}\big)$. Identifying $\mathcal{I}^-$ and $\mathcal{I}^+$ with $S^n,$ we have that the scattering map extends as a Banach space isomorphism on $H^{s+n}(S^n)\times H^{s}(S^n)$ for any $s\geq1$.
\end{enumerate}
\end{theorem}

The same argument also applies in the case of the generalized de Sitter space, where we replace the standard sphere $\big(S^n,g_{S^n}\big)$ by any compact Riemannian manifold $\big(M^n,g_{M}\big)$ which satisfies the equation $Ric(g_M)=(n-1)g_M.$ Additionally, our proof uses robust methods which are generalized to the Einstein equations in upcoming work, in order to prove the nonlinear scattering theory of asymptotically de Sitter spaces in all even spatial dimensions (see already Remark \ref{remark about future work}).

We remark that the scattering problem for the Klein-Gordon equation on asymptotically de Sitter-like spaces was previously addressed in \cite{microlocal}. Subject to a certain relation between the Klein-Gordon mass and the spatial dimension, which in particular does not apply to (\ref{wave equation de sitter intro}), \cite{microlocal} provides a detailed description of the scattering map as a Fourier integral operator. In the case of (\ref{wave equation de sitter intro}), the results of \cite{microlocal} prove that the scattering map is an isomorphism on $C^{\infty}(S^n)\times C^{\infty}(S^n)$. On the other hand, in the present paper we take a different approach to construct the scattering map as a Banach space isomorphism on $H^{s+n}(S^n)\times H^{s}(S^n)$. We also remark that the correspondence between Cauchy data on a finite time slice and asymptotic data at infinity for equation (\ref{wave equation de sitter intro}) was studied in the more general context of linear wave equations on cosmological metric backgrounds by Ringström in \cite{Ringstrom2}. While this work deals with far more complex examples, applied to our situation it only controls $H^s$-type Sobolev norms with an $\epsilon$-loss of derivatives in each direction.

\begin{remark}
    As stated in Theorem \ref{main theorem intro}, we first construct the scattering map in the smooth case. Moreover, we prove the following quantitative estimate for some constant $C>0$:
    \[\big\|\underline{\phi_0}\big\|_{H^{s+n}(S^n)}+\big\|\underline{\mathfrak{h}}\big\|_{H^{s}(S^n)}\leq C\big\|\phi_0\big\|_{H^{s+n}(S^n)}+C\big\|\mathfrak{h}\big\|_{H^{s}(S^n)}\]
    By density, we extend the scattering map to asymptotic data in $H^{s+n}(S^n)\times H^{s}(S^n)$ and obtain a Banach space isomorphism. Using standard density arguments, one can also extend the existence of scattering states, uniqueness of scattering states, and the asymptotic completeness statements to solutions of (\ref{wave equation de sitter intro}) with finite regularity, obtained as a limit of smooth solutions.
\end{remark}

\subsection{Idea of the Proof}

We briefly outline some of the most important ingredients of our proof:
\subsubsection*{Self-similar Solutions of the Wave Equation in Minkowski Space}

The first point is that de Sitter space $\mathbb{R}\times S^n$ can be embedded as a hyperboloid in Minkowski space $\mathbb{R}^{n+2}$, see \cite{HawkingEllis}. Moreover, this hyperboloid is obtained as the quotient of the region $\{u<0,\ v>0\}\subset\mathbb{R}^{n+2}$ by the action of the scaling vector field $S$. As a result, the study of (\ref{wave equation de sitter intro}) is equivalent to studying self-similar solutions of:
\begin{equation}\label{wave equation minkowski intro}
    \square\phi=0
\end{equation}
in the region $\{u<0,\ v>0\}\subset\mathbb{R}^{n+2}$. 

The advantage of this perspective is that we can identify $\mathcal{I}^-=\{T=-\infty\}$ to $\{v=0\}\subset\mathbb{R}^{n+2}$, and similarly $\mathcal{I}^+=\{T=\infty\}$ to $\{u=0\}\subset\mathbb{R}^{n+2}.$ The goal is to prove a scattering theory with asymptotic data at $\{v=0\}$ and $\{u=0\}$. We can interpret this as a compactification which is useful even if the solution is not smooth up to the null cone, unlike the conformal compactification introduced previously. This allows us to reduce a global problem with data at infinity to a finite problem with singular data.

Another advantage is that self-similar vacuum spacetimes are well-studied in the region where $S$ is spacelike in \cite{selfsimilarvacuum}. The paper develops a theory to explain the notion of asymptotic data on the incoming null cone $\{v=0\},$ and it makes rigorous the Fefferman-Graham expansions of \cite{Fefferman-Graham}. In our simplified setting of studying (\ref{wave equation de sitter intro}), this suggests an approach  to define asymptotic data at $\mathcal{I}^-$ and make the expansion (\ref{expansion for phi intro}) rigorous. Moreover, the similarities with \cite{selfsimilarvacuum} justify our notation for $\mathcal{O}$ and $h$. Finally, we point out that the Fefferman-Graham expansion at  $\mathcal{I}^-$ also appears in \cite{Anderson} in the odd dimensional case.

\subsubsection*{The Main Model Equation}

The most important part of the analysis is not carried out at the level of (\ref{wave equation minkowski intro}), but only once we differentiate this $\frac{n}{2}$ times in order to obtain the main model equation.  The scattering theory for (\ref{wave equation minkowski intro}), and implicitly the proof of Theorem \ref{main theorem intro}, will follow from the scattering theory for the main model equation.

To further explain this aspect, we notice that under the above correspondence we have that $\phi$ satisfies the formal expansion on $\{u=-1\}:$
\begin{equation}\label{expansion for phi on u=-1 intro}
    \phi(v)=\sum_{k=0}^{n/2-1}\frac{1}{k!}\phi_k\cdot v^k+\frac{1}{2(n/2)!}\mathcal{O}\cdot v^{\frac{n}{2}}\log v+\frac{1}{(n/2)!}\Tilde{h}\cdot v^{\frac{n}{2}}+O\big(v^{\frac{n}{2}+1}|\log v|^2\big)
\end{equation}
The difficulty in this is that the freely prescribed data is at orders $0$ and $\frac{n}{2}.$ We define $\alpha=\partial_v^{\frac{n}{2}}\phi,$ which satisfies the formal expansion on $\{u=-1\}:$
\begin{equation}\label{expansion for alpha on u=-1 intro}
    \alpha(v)=\frac{1}{2}\mathcal{O}\log v+h+O\big(v|\log v|^2\big)
\end{equation}
where we renormalized $\Tilde{h}$ by a linear factor of $\mathcal{O}.$ We also set $\chi=\partial_v^{\frac{n}{2}-1}\phi$ and we introduce the new time variable $\tau=\sqrt{v}.$ We obtain that $(\alpha,\chi)$ satisfy the main model equation:
\begin{equation}\label{main equation intro}
    \partial_{\tau}^2\alpha+\frac{1}{{\tau}}\partial_{\tau}\alpha+4q'\alpha-\frac{4}{({\tau}^2+1)^2}\cdot\Delta\alpha=f_1(\tau)\tau\partial_{\tau}\alpha+f_2(\tau)\tau^2\alpha+f_3(\tau)\chi
\end{equation}
\begin{equation}\label{definition of chi intro}
    \partial_{\tau}\chi(\tau)=2\tau\alpha(\tau)
\end{equation}
where $|f_1|+|f_2|+|f_3|=O(1)$, $q'\geq1,$ and $\Delta$ represents the Laplacian of the standard metric on $S^n$. The asymptotic data for the system is given by $\big(\chi(0),\mathcal{O},h\big).$

We use the system (\ref{main equation intro})-(\ref{definition of chi intro}) to model the equation for $\partial_v^{\frac{n}{2}}\phi$ along $u=-1$ for $v\in [0,1]$, and similarly the equation for $\partial_u^{\frac{n}{2}}\phi$  along $v=1$ for $u\in [-1,0]$. We then recover the properties of $\phi$ by integration and we use the compatibility relation $\mathcal{O}\sim\Delta^{\frac{n}{2}}\phi_0+\cdots$ to obtain the top order estimates for $\phi.$ This proves the desired scattering statement for self-similar solutions of $(\ref{wave equation minkowski intro})$.

\subsubsection*{Scattering Theory for the Main Model Equation}

Based on the argument outlined above, the essential step in proving Theorem \ref{main theorem intro} is the following scattering result for the main model equation (\ref{main equation intro})-(\ref{definition of chi intro}):

\begin{theorem}[Scattering Theory for the Model Equation]\label{scattering theorem for main equation intro} 
    We have a complete scattering theory for (\ref{main equation intro})-(\ref{definition of chi intro}):

    \begin{enumerate}
    \item Existence and uniqueness of scattering states: there exists a constant $C>0$ such that for any smooth functions $\chi(0),\mathcal{O},h\in C^{\infty}(S^n)$, there exists a unique smooth solution  $\big(\alpha,\chi\big)$ of (\ref{main equation intro})-(\ref{definition of chi intro}) with asymptotic data at $\{\tau=0\}$ given by $\big(\chi(0),\mathcal{O},h\big)$, which satisfies the estimate:
    \begin{equation}\label{forward estimate intro}
         \bigg(\big\|\alpha\big\|_{H^{s+1/2}}+ \big\|\partial_{\tau}\alpha\big\|_{{H^{s-1/2}}}+\big\|\chi\big\|_{H^{s+1/2}}\bigg)\bigg|_{\tau=1}\leq C\bigg(\big\|\mathfrak{h}\big\|_{{H^{s}}}+\big\|\mathcal{O}\big\|_{{H^{s}}}+\big\|\chi(0)\big\|_{H^{s+1/2}}\bigg),
    \end{equation}
    where $\mathfrak{h}:=h-(\log\nabla)\mathcal{O}.$
    \item Asymptotic completeness: there exists a constant $C>0$ such that any smooth solution $\big(\alpha,\chi\big)$ of (\ref{main equation intro})-(\ref{definition of chi intro}) with initial data at $\{\tau=1\}$ given by $\chi(1),\alpha(1),\partial_{\tau}\alpha(1)\in C^{\infty}(S^n)$ induces smooth asymptotic data at $\{\tau=0\}$ given by $\chi(0),\mathcal{O},$ and $h$, which satisfies the estimate:
    \begin{equation}\label{backward estimate intro}
         \big\|\mathfrak{h}\big\|_{{H^{s}}}+\big\|\mathcal{O}\big\|_{{H^{s}}}+ \big\|\chi(0)\big\|_{H^{s+1/2}}\leq C\bigg(\big\|\alpha\big\|_{H^{s+1/2}}+\big\|\partial_{\tau}\alpha\big\|_{{H^{s-1/2}}}+\big\|\chi\big\|_{H^{s+1/2}}\bigg)\bigg|_{\tau=1}
    \end{equation}
    \item The scattering isomorphism: we define the scattering map $\big(\chi(0),\mathcal{O},\mathfrak{h}\big)\mapsto\big(\chi(1),\alpha(1),\partial_{\tau}\alpha(1)\big)$. This extends as a Banach space isomorphism from $H^{s+\frac{1}{2}}(S^n)\times H^s(S^n)\times H^s(S^n)$ to $H^{s+\frac{1}{2}}(S^n)\times H^{s+\frac{1}{2}}(S^n)\times H^{s-\frac{1}{2}}(S^n)$ for any $s\geq1$.
    \end{enumerate}
\end{theorem}

As in the case of Theorem \ref{main theorem intro},the scattering map is defined in the smooth case at first and extended by density using the quantitative estimates (\ref{forward estimate intro}) and (\ref{backward estimate intro}). One remarkable aspect of the above statement is that at nonzero times the solution is bounded in an improved Sobolev space compared to the asymptotic data. We point out that a similar phenomenon is present in the more general context of linear wave equations on cosmological metric
backgrounds in \cite{Ringstrom2}. Restricted to our setting, the paper obtains an almost $\frac{1}{2}$ improvement of regularity in terms of $H^s$-type Sobolev norms, but has an $\epsilon$-loss of derivatives in each direction (which is natural since the quantity corresponding to $h$ is not renormalized).

\begin{remark}
    We point out that equation (\ref{main equation intro}) also models the wave equation on the FLRW spacetime:
    \[-dt^2+t^{\frac{2}{3}}\big(dx_1^2+dx_2^2+dx_3^2\big)\]
    where we set $\tau=t^{\frac{2}{3}}.$ The proof of Theorem \ref{scattering theorem for main equation intro} implies the existence of a scattering isomorphism between Cauchy data at $\tau=1$ and asymptotic initial data at $\tau=0.$ This result was significantly extended to the case of the wave equation and the linearized Einstein-scalar field system on Kasner spacetimes in \cite{Li}.
\end{remark}

\begin{remark}\label{remark about future work}
    The methods used in the proof of Theorem \ref{scattering theorem for main equation intro} can be generalized to prove the scattering of asymptotically de Sitter solutions to the Einstein vacuum equations with a positive cosmological constant in all even spatial dimensions $n\geq4$, to appear in upcoming work. According to \cite{Fefferman-Graham}, \cite{selfsimilarvacuum}, any $(n+1)$-dimensional solution of (\ref{EVE}) corresponds to an $(n+2)$-dimensional straight self-similar vacuum spacetime. In the gauge of \cite{selfsimilarvacuum}, we consider the system of Bianchi equations satisfied by the curvature components of the $(n+2)$-dimensional spacetime along an outgoing null hypersurface. This system generalizes the equation (\ref{main equation intro}) by replacing the round metric on $S^n$ with a time dependent metric. The argument for obtaining energy estimates outlined below is robust and can be adapted to the more general setting by using a geometric definition of the  Littlewood-Paley projections. 
\end{remark}

\subsubsection*{Estimates from $\mathcal{I}^-$ to a Finite Time Hypersurface}

We briefly explain how to obtain the estimate (\ref{forward estimate intro}). In the case of the wave equation (\ref{wave equation de sitter intro}), this gives an estimate on $\frac{n}{2}$ time derivatives of $\Tilde{\phi}$ at finite times in terms of the asymptotic data at $\mathcal{I}^-$. Denoting by $\{\varphi_i\}$ the eigenfunctions of the Laplacian on $S^n$ with eigenvalues $\lambda_i$, we have the frequency decomposition $\alpha=\sum_i\langle\alpha,\varphi_{i}\rangle\varphi_{i}.$ To simplify the discussion, we assume that the solution is supported on frequencies in the interval $[2^l,2^{l+1}),$ so $\langle\alpha,\varphi_{i}\rangle=0$ for all $\lambda_i\notin[2^l,2^{l+1}).$ Also, we introduce the new time variable $t=2^l\tau.$

We decompose the solution into the regular and singular components which satisfy the expansions:
\[\alpha_J(t)=h-\mathcal{O}\log(2^l)+O\big(t^2|\log(t)|^2\big),\]
\[\alpha_Y(t)=\mathcal{O}\log(t)+O\big({t}^2|\log(t)|^2\big).\]
This decomposition highlights the need to define $\mathfrak{h}:=h-(\log\nabla)\mathcal{O}.$ We choose this notation since $\big(\alpha_J,\alpha_Y\big)$ will be shown to satisfy similar bounds as the first and second Bessel functions $J_0,\ Y_0$, as defined in \cite{Bessel}. 

The time interval $t\in[0,1]$ represents the low frequency regime for both components of the solution, with dominant behavior given by the first term in the above expansions. We capture this by energy estimates using multipliers suitable for each component. The time interval $t\in[1,2^{l+1}]$ represents the high frequency regime for the solution, with leading behavior given by $1/\sqrt{t},$ which is again captured using energy estimates. We point out that the high frequency regime behavior, together with the frequency dependent time of transition between the low and high regime, is responsible for the gain of regularity seen in (\ref{forward estimate intro}). This argument is robust, by the use of energy estimates, multipliers, and the time rescaling properties of the equation.

\subsubsection*{Estimates from a Finite Time Hypersurface to $\mathcal{I}^+$}

We outline the proof of the estimate (\ref{backward estimate intro}). In the case of the wave equation (\ref{wave equation de sitter intro}), this implies an estimate on the asymptotic data at $\mathcal{I}^+$ in terms of $\frac{n}{2}$ time derivatives of $\Tilde{\phi}$ at a finite time hypersurface. We assume as before that our solution is supported on frequencies in the interval $[2^l,2^{l+1})$ and we introduce the time variable $t=2^l\tau.$ We study the high frequency regime $t\in[1,2^{l+1}]$ and the low frequency regime $t\in[0,1]$ using different energy estimates than in the previous case, in order to have good bulk terms. We remark that we use the decomposition into frequency strips for constructing multipliers.

A similar strategy of considering two separate regimes also appears in \cite{Ringstrom2}, which studies the ODEs obtained by projecting the solution on each eigenfunction. However, we can construct the desired scattering isomorphism in our situation  since we anticipate the need to renormalize $h$ to $\mathfrak{h}$ before determining the asymptotics.

\subsection{Outline of the Paper}
We outline the remainder of the paper. In Section \ref{set up section}, we introduce the necessary framework by proving the correspondence between solutions of the wave equation on de Sitter space and self-similar solutions of the wave equation in Minkowski space. We also compute the reduction to the main model system (\ref{main equation intro})-(\ref{definition of chi intro}). In Section \ref{main equation estimates section}, we prove the scattering theory result in Theorem \ref{scattering theorem for main equation intro} for the main model system (\ref{main equation intro})-(\ref{definition of chi intro}). In Section \ref{minkowski space scattering section}, we prove the scattering theory result in Theorem \ref{scattering map in Minkowski theorem} for self-similar solutions of the wave equation in Minkowski space, which is based on Theorem \ref{scattering theorem for main equation intro} and the compatibility relations. In Section \ref{scattering in de sitter section}, we complete the proof of Theorem \ref{main theorem intro}, as a consequence of Theorem \ref{scattering map in Minkowski theorem} and the correspondence in Section \ref{set up section}.

\subsection{Acknowledgements}
The author would like to acknowledge Igor Rodnianski for his valuable advice and guidance in the process of writing this paper. The author would also like to thank Mihalis Dafermos and Warren Li for the very helpful discussions and suggestions. 

\section{Set-up}\label{set up section}

In this section we provide the framework needed in our proof. We introduce our coordinate conventions and the correspondence between solutions of the wave equation on de Sitter space and self-similar solutions of the wave equation in Minkowski space. Finally, we derive the main model system of equations (\ref{main equation intro})-(\ref{definition of chi intro}).

\subsection{Coordinate Systems and Self-similarity}\label{change of coordinates subsection}

In the previous section we introduced the de Sitter metric in standard coordinates on $\mathbb{R}\times S^n$:
\[g_{dS}=-dT^2+\cosh^2Td\sigma_n^2\]
where $d\sigma_n^2$ denotes the standard metric on $S^n.$ We want to write the de Sitter space as a quotient of a region in Minkowksi space by the action of the scaling vector field $S$.

We recall that the Minkowski metric in standard coordinates on $\mathbb{R}^{n+2}$ is given by:
\[m=-dt^2+dr^2+r^2d\sigma_n^2\]
Moreover, this metric is self-similar, with the scaling vector field:
\[S=t\partial_t+r\partial_r\]

We introduce the standard double null coordinates:
\[u=\frac{t-r}{2},\ v=\frac{t+r}{2}\]
With respect to the double null coordinates, the Minkowski metric and the scaling vector field are:
\[m=-2\big(du\otimes dv+dv\otimes du\big)+r^2d\sigma_n^2\]
\[S=u\partial_u+v\partial_v\]

In the region $\{u<0,\ v>0\}\subset\mathbb{R}^{n+2}$ we define the self-similar coordinates:
\[x=2\sqrt{-uv},\ T=\frac{1}{2}\log\frac{v}{-u}\]
In self-similar coordinates we have that:
\[m=dx^2+x^2\big(-dT^2+\cosh^2Td\sigma_n^2\big)\]
\[S=x\partial_x\]

The above formula shows that de Sitter space is the quotient of $\{u<0,\ v>0\}\subset\mathbb{R}^{n+2}$ by the action of the scaling vector field $S$. We define the embedding map $\iota:\mathbb{R}\times S^n\rightarrow(0,\infty)\times\mathbb{R}\times S^n$ by $\iota(T,\omega)=(1,T,\omega)$. We also define the projection map $\pi:(0,\infty)\times\mathbb{R}\times S^n\rightarrow\mathbb{R}\times S^n$ by $\pi(x,T,\omega)=(T,\omega)$. We proved that:
\begin{lemma}\label{correspondence lemma}
    $\Tilde{\phi}:\mathbb{R}\times S^n\rightarrow\mathbb{R}$ is a solution of $\square_{dS}\Tilde{\phi}=0$ if and only if $\phi:(0,\infty)\times\mathbb{R}\times S^n\rightarrow\mathbb{R}$ given by $\phi=\Tilde{\phi}\circ\pi$ is a self-similar solution of $\square\phi=0.$ If this holds, we also have $\Tilde{\phi}=\phi\circ\iota.$
\end{lemma}

We conclude that studying the linear wave equation on de Sitter space $\mathbb{R}\times S^n$ is equivalent to studying self-similar solutions of the wave equation on the $\{u<0,\ v>0\}$ region of Minkowski space $\mathbb{R}^{n+2}$. Moreover, we can identify $\mathcal{I}^-=\{T=-\infty\}$ to $\{v=0\}\subset\mathbb{R}^{n+2}$, and similarly $\mathcal{I}^+=\{T=\infty\}$ to $\{u=0\}\subset\mathbb{R}^{n+2}.$

\subsection{Derivation of the Main Model Equation}\label{deriving model problem}

For any even integer $n\geq 4$, we have the wave equation in double null coordinates on $\mathbb{R}^{n+2}:$
\begin{equation}\label{wave equation phi}
    \partial_u\partial_v\phi-\frac{n/2}{v-u}\partial_v\phi+\frac{n/2}{v-u}\partial_u\phi-(v-u)^{-2}\Delta\phi=0
\end{equation}
where $\Delta$ represents the Laplacian of the standard metric on $S^n.$

As explained in the introduction, we want to derive an equation for $\partial_v^{\frac{n}{2}}\phi$ on $u=-1.$ We use the self-similarity assumption $S\phi=0$ in order to rewrite the above wave equation on $u=-1:$
\begin{equation}\label{wave equation phi u=-1}
    v\partial_v^2\phi+\bigg(1-\frac{n}{2}\bigg)\partial_v\phi+\frac{nv}{v+1}\partial_v\phi-\frac{1}{(v+1)^2}\Delta\phi=0
\end{equation}
We multiply by $(v+1)^2$ and differentiate with respect to $v$:
\begin{equation}\label{wave equation dv3 phi}
    v(v+1)^2\partial_v^3\phi+\bigg(2-\frac{n}{2}\bigg)(v+1)^2\partial_v^2\phi+(n+2)(v+1)v\partial_v^2\phi+(nv+2v+2)\partial_v\phi-\Delta\partial_v\phi=0
\end{equation}
By induction we obtain that for all $2\leq a\leq\frac{n}{2},$ there exist constants $p_a,q_a,p_a',q_a',q_a''\geq 1$:
\begin{align}\label{wave equation dva phi}
    v(v+1)^2\partial_v^{a+2}\phi &+\bigg(a+1-\frac{n}{2}\bigg)(v+1)^2\partial_v^{a+1}\phi+\big[n(v+1)+p_av+q_a\big]v\partial_v^{a+1}\phi+
    \\ \nonumber &+(p'_av+q'_a)\partial_v^a\phi+q_a''\partial_v^{a-1}\phi-\Delta\partial_v^a\phi=0
\end{align}

In particular, in the case $a=\frac{n}{2}$ we obtain the desired equation for $\partial_v^{\frac{n}{2}}\phi$:
\[v\partial_v^{\frac{n}{2}+2}\phi+\partial_v^{\frac{n}{2}+1}\phi+\frac{pv+q}{(v+1)^2}v\partial_v^{\frac{n}{2}+1}\phi+\frac{p'v+q'}{(v+1)^2}\partial_v^\frac{n}{2}\phi+\frac{q''}{(v+1)^2}\partial_v^{\frac{n}{2}-1}\phi-\frac{1}{(v+1)^2}\Delta\partial_v^\frac{n}{2}\phi=0\]
for some constants $p,q,p',q',q''$ with $q'\geq1.$ 

We introduce the notation $\alpha=\partial_v^{\frac{n}{2}}\phi,\ \chi=\partial_v^{\frac{n}{2}-1}\phi$. Thus, we have deduced the following equation:
\begin{equation}\label{alpha}
    v\partial_v^2\alpha+\partial_v\alpha+\frac{pv+q}{(v+1)^2}v\partial_v\alpha+\frac{p'v+q'}{(v+1)^2}\alpha+\frac{q''}{(v+1)^2}\chi-\frac{1}{(v+1)^2}\Delta\alpha=0
\end{equation}
Finally, under the change of variables $\tau=\sqrt{v}$ we obtain:
\[\partial_{\tau}^2\alpha+\frac{1}{{\tau}}\partial_{\tau}\alpha+4q'\alpha-\frac{4}{({\tau}^2+1)^2}\cdot\Delta\alpha= f_1(\tau)\tau\partial_{\tau}\alpha+f_2(\tau)\tau^2\alpha+f_3(\tau)\chi\]
\[\partial_{\tau}\chi(\tau)=2\tau\alpha(\tau)\]
where $|f_1|+|f_2|+|f_3|=O(1)$ can be computed explicitly, but their exact value is irrelevant. This completes the derivation of the main model system of equations (\ref{main equation intro})-(\ref{definition of chi intro}).

\section{Estimates for the Main Model Equation}\label{main equation estimates section}

The essential part of the analysis needed for the scattering theory for the wave equation on de Sitter space in Theorem \ref{main theorem intro} is carried out at the level of (\ref{main equation intro})-(\ref{definition of chi intro}). In this section, we study this model system of equations:
\begin{equation}\label{main equation}
    \partial_{\tau}^2\alpha+\frac{1}{{\tau}}\partial_{\tau}\alpha+4q'\alpha-\frac{4}{({\tau}^2+1)^2}\cdot\Delta\alpha=f_1(\tau)\tau\partial_{\tau}\alpha+f_2(\tau)\tau^2\alpha+f_3(\tau)\chi
\end{equation}
\begin{equation}\label{definition of chi}
    \partial_{\tau}\chi(\tau)=2\tau\alpha(\tau)
\end{equation}
where $|f_1|+|f_2|+|f_3|=O(1),$ $q'\geq1$, and $\Delta$ represents the Laplacian of the standard metric on $S^n$.

We first introduce the notion of a smooth solution of the above system:

\begin{definition}
    Let $\chi(0),\mathcal{O},h\in C^{\infty}(S^n).$ We say that $\big(\alpha,\chi\big)$ is a smooth solution of (\ref{main equation})-(\ref{definition of chi}) with asymptotic initial data given by $\chi(0),\mathcal{O},$ and $h$ if: 
    \[\alpha-\mathcal{O}\log({\tau})-h\in C^1_{\tau}([0,\infty))C^{\infty}(S^n),\ \chi\in C^1_{\tau}([0,\infty))C^{\infty}(S^n)\] and $\alpha$ satisfies the expansions:
    \begin{equation}\label{main expansion}
        \alpha({\tau})=\mathcal{O}\log({\tau})+h+ O\big({\tau}^2|\log({\tau})|^2\big),\ \partial_{\tau}\alpha({\tau})=\frac{\mathcal{O}}{\tau}+O\big({\tau}|\log({\tau})|^2\big)\text{ in }C^{\infty}(S^n)
    \end{equation}
    
    Given $\alpha,\chi\in C^{\infty}\big((0,\infty)\times S^n\big)$ solving (\ref{main equation})-(\ref{definition of chi}), we say that $\big(\chi(0),\mathcal{O},h\big)$ determine the asymptotic data of the solution if the above conditions hold.
\end{definition}

The main result of the section is the proof of Theorem \ref{scattering theorem for main equation intro}. This consists of proving energy estimates from $\mathcal{I}^-$ to finite times in Section \ref{forward direction subsection}, the existence and uniqueness of scattering states in Section \ref{existence to ivp subsection}, energy estimates from finite times to $\mathcal{I}^+$ in Section \ref{backward direction subsection}, and asymptotic completeness in Section \ref{asymptotic expansion subsection}. As a result, we can construct the scattering map from asymptotic data at $\tau=0$ to initial data at $\tau=1$, and we obtain it is an isomorphism.

\subsection{Estimates from $\mathcal{I}^-$ to a Finite Time Hypersurface}\label{forward direction subsection}

In this section we assume the existence of smooth solutions of (\ref{main equation})-(\ref{definition of chi}) with prescribed asymptotic initial data, and prove estimates on the solutions at nonzero times $\tau\in(0,1]$ in terms of the asymptotic data. In the context of (\ref{wave equation de sitter intro}), these correspond to estimates on $\frac{n}{2}$ time derivatives of $\Tilde{\phi}$ at finite times in terms of the asymptotic data at $\mathcal{I}^-$. The main result in Theorem \ref{improved regularity theorem} proves the estimate (\ref{forward estimate intro}) of Theorem \ref{scattering theorem for main equation intro}. As a consequence of our estimates, we also obtain the uniqueness of solutions with given smooth asymptotic initial data (\textit{uniqueness of scattering states}). 

Let $\chi(0),\mathcal{O},h$ be smooth functions. For any parameters $m_J,m_Y\geq1,$ we consider $\big(\alpha_J,\chi_J\big)$ and $\big(\alpha_Y,\chi_Y\big)$ to be smooth solutions of (\ref{main equation})-(\ref{definition of chi}) with asymptotic initial data given by $\big(\chi_J(0)=\frac{1}{2}\chi(0),0,\mathfrak{h}\big)$ and $\big(\chi_Y(0)=\frac{1}{2}\chi(0),\mathcal{O},\mathcal{O}\log(m_Y)\big),$ where $\mathfrak{h}=h-\mathcal{O}\log(m_Y).$ In particular, the solutions satisfy the expansions:
\[\alpha_J({\tau})=\mathfrak{h}+O\big({\tau}^2|\log({\tau})|^2\big),\]
\[\alpha_Y({\tau})=\mathcal{O}\log(m_Y{\tau})+O\big({\tau}^2|\log({\tau})|^2\big).\]
Using the fact that (\ref{main equation}) and (\ref{definition of chi}) are linear, we obtain that $\alpha=\alpha_J+\alpha_Y,\ \chi=\chi_J+\chi_Y$ also solve (\ref{main equation})-(\ref{definition of chi}), with asymptotic initial data given by $\big(\chi(0),\mathcal{O},h\big).$ We refer to $\alpha_J$ as the regular component of $\alpha$, and to $\alpha_Y$ as the singular component of $\alpha$. We recall that our notation suggests the fact that $\alpha_J$ and $\alpha_Y$ satisfy similar bounds as the first and second Bessel functions.

When proving estimates for the solution, we treat separately the low frequency regime $\tau\in[0,(2m)^{-1}]$ with dominant behavior given by the first term in the above expansions, and the high frequency regime $\tau\in[(2m)^{-1},1]$ with dominant behavior bounded by $1/\sqrt{m\tau}.$ We remark that the transition time depends on our choice of the parameter $m.$

We also remark that in order to prove estimates for the two components of the solution, we use different multipliers adapted to their asymptotic expansion at $\tau=0$. We start by proving a low frequency regime estimate for the regular component of the solution:

\begin{proposition}\label{low frequency alpha J proposition} For any $\tau\leq(m_J)^{-1}\leq1$, we have that $\alpha_J$ and $\chi_J$ satisfy the estimates:
    \begin{equation}\label{alpha J in low frequency 1}
    \big\|\alpha_J\big\|^2_{H^1}+\big\|\partial_{\tau}\alpha_J\big\|^2_{L^2} \lesssim\big\|\mathfrak{h}\big\|^2_{H^1}+\frac{1}{m_J^2}\big\|\chi(0)\big\|^2_{L^2}
    \end{equation}
    \begin{equation}\label{alpha J in low frequency 2}
        \big\|\alpha_J\big\|^2_{L^2}\lesssim \big\|\mathfrak{h}\big\|^2_{L^2}+\tau^2\big\|\nabla\mathfrak{h}\big\|^2_{L^2}+\frac{\tau^2}{m_J^2}\big\|\chi(0)\big\|^2_{L^2}
    \end{equation}
    \begin{equation}\label{chi J in low frequency}
        \big\|\chi_J\big\|^2_{L^2}\lesssim \big\|\chi(0)\big\|^2_{L^2}+\tau^4\big\|\mathfrak{h}\big\|^2_{H^1}
    \end{equation}
    where the implicit constant in the above inequalities is independent of $m_J, \tau,$ and the asymptotic data.
\end{proposition}
\begin{proof}
    We introduce the new time variable $t=m_J\tau\in[0,1]$. Equation (\ref{main equation}) can be written as:
    \[\partial_{t}^2\alpha_J+\frac{1}{t}\partial_{t}\alpha_J+4q'\frac{\alpha_J}{m_J^2}-\frac{4}{(t^2/m_J^2+1)^2}\cdot\Delta\frac{\alpha_J}{m_J^2}=\frac{f_1'(t)}{m_J^2}t\partial_{t}\alpha_J+\frac{f_2'(t)}{m_J^4}t^2\alpha_J+\frac{f_3'(t)}{m_J^2}\chi_J\]
    where $f'_1,f'_2,f'_3$ are bounded functions of $t.$ We also notice that from (\ref{definition of chi}) we have:
    \[\chi_J(t)=\chi(0)+\frac{1}{m_J^2}\int_0^t2t'\alpha_J(t')dt'\]
    In the following inequalities, the implicit constant will be independent of $m_J, t,$ and the asymptotic data. Multiplying the equation by $\partial_t\alpha_J$ and integrating by parts, we obtain the standard energy estimate:
    \[\big\|\partial_{t}\alpha_J\big\|^2_{L^2}(t)+\int_0^t\frac{1}{t'}\big\|\partial_{t}\alpha_J\big\|^2_{L^2}(t')dt'+\frac{1}{m_J^2}\big\|\alpha_J\big\|^2_{H^1}(t)+\frac{1}{m_J^4}\int_0^tt'\big\|\nabla\alpha_J\big\|^2_{L^2}(t')dt'\lesssim\]\[\lesssim\frac{1}{m_J^2}\big\|\mathfrak{h}\big\|^2_{H^1}+\int_0^t\frac{t'}{m_J^2}\big\|\partial_{t}\alpha_J\big\|^2_{L^2}+\int_0^t\int_S\frac{1}{m_J^2}\bigg(|\chi_J|+\frac{t'^2}{m_J^2}|\alpha_J|\bigg)\cdot|\partial_t\alpha_J|dt'\]
    We notice that for $t\in\big[0,1\big],$ we can apply Gronwall to $\partial_t\alpha$ to obtain:
    \[\big\|\partial_{t}\alpha_J\big\|^2_{L^2}+\frac{1}{m_J^2}\big\|\alpha_J\big\|^2_{H^1}\lesssim\frac{1}{m_J^2}\big\|\mathfrak{h}\big\|^2_{H^1}+\int_0^t\frac{1}{m_J^4}\bigg(\big\|\chi_J\big\|^2_{L^2}+\frac{1}{m_J^4}\big\|\alpha_J\big\|^2_{L^2}\bigg)\]
    We also have the bound:
    \[\big\|\chi_J\big\|^2_{L^2}\lesssim\big\|\chi(0)\big\|^2_{L^2}+\frac{1}{m_J^4}\int_0^t\big\|\alpha_J\big\|^2_{L^2}dt'\]
    We combine our previous two estimates to obtain:
    \[\big\|\partial_{t}\alpha_J\big\|^2_{L^2}+\frac{1}{m_J^2}\cdot\big\|\alpha_J\big\|^2_{H^1}\lesssim\frac{1}{m_J^2}\big\|\mathfrak{h}\big\|^2_{H^1}+\frac{1}{m_J^4}\big\|\chi(0)\big\|^2_{L^2}+\frac{1}{m_J^8}\int_0^t\big\|\alpha_J\big\|^2_{L^2}dt'\]
    We apply Gronwall to obtain (\ref{alpha J in low frequency 1}). We use this in the inequality:
    \[\big\|\alpha_J\big\|_{L^2}\lesssim \big\|\mathfrak{h}\big\|_{L^2}+\int_0^{\tau}\big\|\partial_{\tau}\alpha_J\big\|_{L^2}d\tau'\]
    in order to prove (\ref{alpha J in low frequency 2}). Finally, these bounds together with (\ref{definition of chi}) imply (\ref{chi J in low frequency}).
\end{proof}

We prove a similar low frequency regime estimate for singular component of the solution:

\begin{proposition}\label{low frequency alpha Y proposition} For any $\tau\leq(2m_Y)^{-1}$, we have that $\alpha_Y$ satisfies the estimate:
    \begin{equation}\label{alpha Y in low frequency 1}
        \bigg\|\frac{\alpha_Y}{\log(m_Y\tau)}\bigg\|^2_{H^1}+\bigg\|\partial_{\tau}\frac{\alpha_Y}{\log(m_Y\tau)}\bigg\|^2_{L^2}\lesssim \big\|\mathcal{O}\big\|^2_{H^1}+\frac{1}{m_Y^2}\big\|\chi(0)\big\|^2_{L^2}
    \end{equation}
    \begin{equation}\label{alpha Y in low frequency 2}
        \bigg\|\frac{\alpha_Y}{\log(m_Y\tau)}\bigg\|^2_{L^2}\lesssim \big\|\mathcal{O}\big\|^2_{L^2}+\tau^2\big\|\nabla\mathcal{O}\big\|^2_{L^2}+\frac{\tau^2}{m_Y^2}\big\|\chi(0)\big\|^2_{L^2}
    \end{equation}
    \begin{equation}\label{chi Y in low frequency}
        \big\|\chi_Y\big\|^2_{L^2}\lesssim \big\|\chi(0)\big\|^2_{L^2}+\tau^2\big\|\mathcal{O}\big\|^2_{H^1}
    \end{equation}
    where the implicit constant in the above inequalities is independent of $m_Y, \tau,$ and the asymptotic data.
\end{proposition}
\begin{proof}
    We introduce the new time variable $t=m_Y\tau\in[0,1/2]$. We notice that the expansion for $\alpha_Y$ implies:
    \[\frac{\alpha_Y}{\log t}\bigg|_{t=0}=\mathcal{O},\ \partial_t\bigg(\frac{\alpha_Y}{\log t}\bigg)\bigg|_{t=0}=0\]
    As before, equation (\ref{main equation}) can be written as:
    \[\partial_{t}^2\alpha_Y+\frac{1}{t}\partial_{t}\alpha_Y+4q'\frac{\alpha_Y}{m_Y^2}-\frac{4}{(t^2/m_Y^2+1)^2}\cdot\Delta\frac{\alpha_Y}{m_Y^2}=\frac{f_1'(t)}{m_Y^2}t\partial_{t}\alpha_Y+\frac{f_2'(t)}{m_Y^4}t^2\alpha_Y+\frac{f_3'(t)}{m_Y^2}\chi_Y\]
    where $f'_1,f'_2,f'_3$ are bounded functions of $t.$ We multiply by $\frac{1}{\log t}$ to get:
    \[\partial_{t}^2\bigg(\frac{\alpha_Y}{\log t}\bigg)+\frac{1}{t}\bigg(1+\frac{2}{\log t}\bigg)\partial_{t}\bigg(\frac{\alpha_Y}{\log t}\bigg)+\frac{4q'}{m_Y^2}\cdot\frac{\alpha_Y}{\log t}-\frac{4}{(t^2/m_Y^2+1)^2}\cdot\frac{1}{m_Y^2}\Delta\frac{\alpha_Y}{\log t}=\]\[=\frac{f_1''(t)}{m_Y^2}t\partial_{t}\bigg(\frac{\alpha_Y}{\log t}\bigg)+\frac{f_2''(t)}{m_Y^2}\cdot\frac{\alpha_Y}{\log t}+\frac{f_3''(t)}{m_Y^2}\chi_Y\]
    where $f''_1,f''_2,f''_3$ are bounded functions of $t.$ We also have from (\ref{definition of chi}):
    \[\chi_Y(t)=\chi(0)+\frac{1}{m_Y^2}\int_0^t2t'\alpha_Y(t')dt'\]
    which implies the bound:
    \[\big\|\chi_Y\big\|^2_{L^2}\lesssim\big\|\chi(0)\big\|^2_{L^2}+\frac{1}{m_Y^4}\int_0^t\bigg\|\frac{\alpha_Y}{\log t}\bigg\|^2_{L^2}dt'\]
    In the following inequalities, the implicit constant will be independent of $m_Y, t,$ and the asymptotic data. Using the equation for $\frac{\alpha_Y}{\log t},$ we obtain the standard energy estimate:
    \[\bigg\|\partial_{t}\frac{\alpha_Y}{\log t}\bigg\|^2_{L^2}(t)+\int_0^t\frac{1}{t'}\bigg\|\partial_{t}\frac{\alpha_Y}{\log t}\bigg\|^2_{L^2}(t')dt'+\frac{1}{m_Y^2}\bigg\|\frac{\alpha_Y}{\log t}\bigg\|^2_{H^1}(t)+\frac{1}{m_Y^4}\int_0^tt'\bigg\|\frac{\nabla\alpha_Y}{\log t}\bigg\|^2_{L^2}(t')dt'\lesssim\]\[\lesssim\frac{1}{m_Y^2}\big\|\mathcal{O}\big\|^2_{H^1}+\int_0^t\bigg(\frac{t'}{m_Y^2}+\frac{\textbf{1}_{[1/10,1/2]}}{t'|\log t'|}\bigg)\bigg\|\partial_{t}\frac{\alpha_Y}{\log t}\bigg\|^2_{L^2}+\int_0^t\int_S\frac{1}{m_Y^2}\bigg(|\chi_Y|+\bigg|\frac{\alpha_Y}{\log t}\bigg|\bigg)\cdot\bigg|\partial_t\frac{\alpha_Y}{\log t}\bigg|\]
    We point out that the error term $\frac{1}{t'|\log t'|}\big\|\partial_{t}\frac{\alpha_Y}{\log t}\big\|^2_{L^2}\cdot\textbf{1}_{[1/10,1/2]}$ appears on the RHS because for $t\in[0,1/10]$ we have $1+2/\log t\gtrsim1.$ Since $t\in\big[0,\frac{1}{2}\big],$ we apply Gronwall to $\partial_t\frac{\alpha_Y}{\log t}$ to obtain:
    \[\bigg\|\partial_{t}\frac{\alpha_Y}{\log t}\bigg\|^2_{L^2}+\frac{1}{m_Y^2}\cdot\bigg\|\frac{\alpha_Y}{\log t}\bigg\|^2_{H^1}(t)\lesssim\frac{1}{m_Y^2}\big\|\mathcal{O}\big\|^2_{H^1}+\frac{1}{m_Y^4}\int_0^t\bigg(\big\|\chi_Y\big\|^2_{L^2}+\bigg\|\frac{\alpha_Y}{\log t}\bigg\|^2_{L^2}\bigg)\lesssim\]
    \[\lesssim\frac{1}{m_Y^2}\big\|\mathcal{O}\big\|^2_{H^1}+\frac{1}{m_Y^4}\big\|\chi(0)\big\|^2_{L^2}+\frac{1}{m_Y^4}\int_0^t\bigg\|\frac{\alpha_Y}{\log t}\bigg\|^2_{L^2}(t')dt'\]
    We apply Gronwall to obtain (\ref{alpha Y in low frequency 1}). We use this in the inequality:
    \[\bigg\|\frac{\alpha_Y}{\log(m_Y\tau)}\bigg\|_{L^2}\lesssim \big\|\mathcal{O}\big\|_{L^2}+\int_0^{\tau}\bigg\|\partial_{\tau}\frac{\alpha_Y}{\log(m_Y\tau)}\bigg\|_{L^2}d\tau'\]
    in order to prove (\ref{alpha Y in low frequency 2}). Finally, these bounds together with (\ref{definition of chi}) imply (\ref{chi Y in low frequency}).
\end{proof}

Next, we prove a high frequency regime estimate which applies both to the regular component and to the singular component of the solution:

\begin{proposition}\label{high frequency forward estimate} For any parameter $m\geq 1,$ and any $\tau\in\big[(2m)^{-1},1\big]$ we have that $\alpha$ satisfies the estimate:
    \[\frac{1}{\tau}\big\|\alpha\big\|^2_{L^2}+\tau\big\|\nabla\alpha\big\|^2_{L^2}+\tau\big\|\partial_{\tau}\alpha\big\|^2_{L^2}\lesssim \frac{1}{m}\cdot\bigg(m^2\big\|\alpha\big\|^2_{L^2}+\big\|\nabla\alpha\big\|^2_{L^2}+\big\|\partial_{\tau}\alpha\big\|^2_{L^2}+m\big\|\chi\big\|^2_{L^2}\bigg)\bigg|_{\tau=(2m)^{-1}}\]
    where the implicit constant in the above inequality is independent of $m$ and $\tau.$
\end{proposition}
\begin{proof}
    We introduce the new time variable $t=m\tau\in[1/2,m]$. As before, equation (\ref{main equation}) can be written as:
    \[\partial_{t}^2\alpha+\frac{1}{t}\partial_{t}\alpha+4q'\frac{\alpha}{m^2}-\frac{4}{(t^2/m^2+1)^2}\cdot\Delta\frac{\alpha}{m^2}=\frac{f_1'(t)}{m^2}t\partial_{t}\alpha+\frac{f_2'(t)}{m^4}t^2\alpha+\frac{f_3'(t)}{m^2}\chi\]
    where $f'_1,f'_2,f'_3$ are bounded functions of $t.$ We multiply by $\sqrt{t}$ to get:
    \[\partial_{t}^2\big(\alpha\sqrt{t}\big)+\frac{1}{4t^2}\alpha\sqrt{t}+\frac{4q'}{m^2}\alpha\sqrt{t}-\frac{4}{(t^2/m^2+1)^2}\cdot\Delta\frac{\alpha\sqrt{t}}{m^2}=\frac{f_1''(t)}{m^2}t\partial_{t}\big(\alpha\sqrt{t}\big)+\frac{f_2''(t)}{m^2}\alpha\sqrt{t}+\frac{f_3''(t)}{m^2}\chi\sqrt{t}\]
    where $f''_1,f''_2,f''_3$ are bounded functions of $t.$ We also have from (\ref{definition of chi}):
    \[\chi(t)=\chi(1/2)+\frac{1}{m^2}\int_{1/2}^t2t'\alpha(t')dt'\]
    which implies the bound for $t\in[1/2,m]$:
    \[\big\|\chi\big\|^2_{L^2}\lesssim\big\|\chi(1/2)\big\|^2_{L^2}+\frac{1}{m^2}\int_{1/2}^t\big\|\alpha\sqrt{t'}\big\|^2_{L^2}dt'\]
    In the following inequalities, the implicit constant will be independent of $m$ and $t.$ Using the equation for $\alpha\sqrt{t},$ we obtain the standard energy estimate:
    \[\big\|\partial_{t}\big(\alpha\sqrt{t}\big)\big\|^2_{L^2}(t)+\frac{1}{t^2}\big\|\alpha\sqrt{t}\big\|^2_{L^2}(t)+\int_{1/2}^t\frac{1}{t'^3}\big\|\alpha\sqrt{t'}\big\|^2_{L^2}dt'+\frac{1}{m^2}\big\|\nabla\alpha\sqrt{t}\big\|^2_{L^2}(t)+\frac{1}{m^4}\int_{1/2}^tt'\big\|\nabla\alpha\sqrt{t'}\big\|^2_{L^2}dt'\lesssim\]\[\lesssim\bigg(\big\|\partial_{t}\alpha\big\|^2_{L^2}+\big\|\alpha\big\|^2_{L^2}+\frac{1}{m^2}\big\|\nabla\alpha\big\|^2_{L^2}\bigg)\bigg|_{t=\frac{1}{2}}+\int_{1/2}^t\frac{t'}{m^2}\big\|\partial_{t}\big(\alpha\sqrt{t'}\big)\big\|^2_{L^2}+\int_{1/2}^t\int_S\frac{\sqrt{t'}}{m^2}\big(|\chi|+|\alpha|\big)\cdot\big|\partial_{t}\big(\alpha\sqrt{t'}\big)\big|\]
    We remark that we can bound:
    \[\int_{1/2}^t\int_S\frac{\sqrt{t'}}{m^2}\big(|\chi|+|\alpha|\big)\cdot\big|\partial_{t}\big(\alpha\sqrt{t'}\big)\big|\lesssim\int_{1/2}^t\frac{t'}{m^2}\big\|\partial_{t}\big(\alpha\sqrt{t'}\big)\big\|^2_{L^2}+\frac{1}{m^2}\int_{1/2}^t\big\|\alpha\big\|^2_{L^2}+\big\|\chi\big\|^2_{L^2}\lesssim\]\[\lesssim\frac{1}{m}\big\|\chi(1/2)\big\|^2_{L^2}+\int_{1/2}^t\frac{t'}{m^2}\big\|\partial_{t}\big(\alpha\sqrt{t'}\big)\big\|^2_{L^2}dt'+\frac{1}{m^2}\int_{1/2}^t\big\|\alpha\big\|^2_{L^2}dt'\]
    Combining the previous two inequalities we have:
    \[\big\|\partial_{t}\big(\alpha\sqrt{t}\big)\big\|^2_{L^2}(t)+\frac{1}{t^2}\big\|\alpha\sqrt{t}\big\|^2_{L^2}(t)+\frac{1}{m^2}\big\|\nabla\alpha\sqrt{t}\big\|^2_{L^2}(t)\lesssim\]\[\lesssim\bigg(\big\|\partial_{t}\alpha\big\|^2_{L^2}+\big\|\alpha\big\|^2_{L^2}+\frac{1}{m^2}\big\|\nabla\alpha\big\|^2_{L^2}+\frac{1}{m}\big\|\chi\big\|^2_{L^2}\bigg)\bigg|_{t=\frac{1}{2}}+\int_{1/2}^t\frac{t'}{m^2}\bigg(\big\|\partial_{t}\big(\alpha\sqrt{t'}\big)\big\|^2_{L^2}+\frac{1}{t'^2}\big\|\alpha\sqrt{t'}\big\|^2_{L^2}\bigg)\]
    Since $t\leq m,$ we can apply Gronwall to obtain:
     \[\big\|\partial_{t}\big(\alpha\sqrt{t}\big)\big\|^2_{L^2}(t)+\frac{1}{t^2}\big\|\alpha\sqrt{t}\big\|^2_{L^2}(t)+\frac{1}{m^2}\big\|\nabla\alpha\sqrt{t}\big\|^2_{L^2}(t)\lesssim\bigg(\big\|\partial_{t}\alpha\big\|^2_{L^2}+\big\|\alpha\big\|^2_{L^2}+\frac{1}{m^2}\big\|\nabla\alpha\big\|^2_{L^2}+\frac{1}{m}\big\|\chi\big\|^2_{L^2}\bigg)\bigg|_{t=\frac{1}{2}}\]
     Finally, we can rewrite this as:
     \[\frac{1}{\tau}\big\|\alpha\big\|^2_{L^2}+\tau\big\|\nabla\alpha\big\|^2_{L^2}+\tau\big\|\partial_{\tau}\alpha\big\|^2_{L^2}\lesssim \frac{1}{m}\cdot\bigg(m^2\big\|\alpha\big\|^2_{L^2}+\big\|\nabla\alpha\big\|^2_{L^2}+\big\|\partial_{\tau}\alpha\big\|^2_{L^2}+m\big\|\chi\big\|^2_{L^2}\bigg)\bigg|_{\tau=(2m)^{-1}}\]
\end{proof}

As previously explained, we combine the low frequency regime and the high frequency regime estimates for the regular component of the solution, and we obtain the following result:

\begin{corollary}\label{forward estimate alpha J}
There exists an implicit constant independent of $m_J$ and the asymptotic data, such that $\alpha_J$ satisfies the estimate:
    \[\bigg(\big\|\alpha_J\big\|^2_{H^1}+\big\|\partial_{\tau}\alpha_J\big\|^2_{L^2}\bigg)\bigg|_{\tau=1}\lesssim m_J\big\|\mathfrak{h}\big\|^2_{L^2}+\frac{1}{m_J}\big\|\nabla\mathfrak{h}\big\|^2_{L^2}+\big\|\chi(0)\big\|^2_{L^2}\]
\end{corollary}
\begin{proof}
    We apply Proposition \ref{high frequency forward estimate} to $\alpha_J,\chi_J,$ and $m_J:$
    \[\bigg(\big\|\alpha_J\big\|^2_{H^1}+\big\|\partial_{\tau}\alpha_J\big\|^2_{L^2}\bigg)\bigg|_{\tau=1}\lesssim \frac{1}{m_J}\bigg(m_J^2\big\|\alpha_J\big\|^2_{L^2}+\big\|\nabla\alpha_J\big\|^2_{L^2}+\big\|\partial_{\tau}\alpha_J\big\|^2_{L^2}+m_J\big\|\chi_J\big\|^2_{L^2}\bigg)\bigg|_{\tau=(2m_J)^{-1}}\]
    Next, we use (\ref{alpha J in low frequency 1}),(\ref{alpha J in low frequency 2}), and (\ref{chi J in low frequency}) to get: 
    \[\big\|\alpha_J\big\|^2_{H^1}\big((2m_J)^{-1}\big)+\big\|\partial_{\tau}\alpha_J\big\|^2_{L^2}\big((2m_J)^{-1}\big)\lesssim \big\|\mathfrak{h}\big\|^2_{H^1}+\frac{1}{m_J^2}\big\|\chi(0)\big\|^2_{L^2}\]
    \[\big\|\alpha_J\big\|^2_{L^2}\big((2m_J)^{-1}\big)\lesssim \big\|\mathfrak{h}\big\|^2_{L^2}+\frac{1}{m_J^2}\big\|\nabla\mathfrak{h}\big\|^2_{L^2}+\frac{1}{m_J^4}\big\|\chi(0)\big\|^2_{L^2}\]
    \[\big\|\chi_J\big\|^2_{L^2}\big((2m_J)^{-1}\big)\lesssim \big\|\chi(0)\big\|^2_{L^2}+\frac{1}{m_J^4}\big\|\mathfrak{h}\big\|^2_{H^1}\]
    Combining all these inequalities we obtain the desired estimate.
\end{proof}

Similarly, we use the low frequency regime and the high frequency regime estimates for the singular component of the solution in order to obtain:

\begin{corollary}\label{forward estimate alpha Y}
There exists an implicit constant independent of $m_Y$ and the asymptotic data, such that $\alpha_Y$ satisfies the estimate:
    \[\bigg(\big\|\alpha_Y\big\|^2_{H^1}+\big\|\partial_{\tau}\alpha_Y\big\|^2_{L^2}\bigg)\bigg|_{\tau=1}\lesssim m_Y\big\|\mathcal{O}\big\|^2_{L^2}+\frac{1}{m_Y}\big\|\nabla\mathcal{O}\big\|^2_{L^2}+\big\|\chi(0)\big\|^2_{L^2}\]
\end{corollary}
\begin{proof}
    We apply Proposition \ref{high frequency forward estimate} to $\alpha_Y,\chi_Y,$ and $m_Y:$
    \[\bigg(\big\|\alpha_Y\big\|^2_{H^1}+\big\|\partial_{\tau}\alpha_Y\big\|^2_{L^2}\bigg)\bigg|_{\tau=1}\lesssim \frac{1}{m_Y}\bigg(m_Y^2\big\|\alpha_Y\big\|^2_{L^2}+\big\|\nabla\alpha_Y\big\|^2_{L^2}+\big\|\partial_{\tau}\alpha_Y\big\|^2_{L^2}+m_Y\big\|\chi_Y\big\|^2_{L^2}\bigg)\bigg|_{\tau=(2m_Y)^{-1}}\]
    We notice that:
    \[\big\|\partial_{\tau}\alpha_Y\big\|^2_{L^2}\big((2m_Y)^{-1}\big)\lesssim m_Y^2\big\|\alpha_Y\big\|^2_{L^2}\big((2m_Y)^{-1}\big)+\bigg\|\partial_{\tau}\frac{\alpha_Y}{\log(m_Y\tau)}\bigg\|^2_{L^2}\big((2m_Y)^{-1}\big)\]
    Next, we use (\ref{alpha Y in low frequency 1}),(\ref{alpha Y in low frequency 2}), and (\ref{chi Y in low frequency}) to get: 
    \[\big\|\alpha_Y\big\|^2_{H^1}\big((2m_Y)^{-1}\big)+\bigg\|\partial_{\tau}\frac{\alpha_Y}{\log(m_Y\tau)}\bigg\|^2_{L^2}\big((2m_Y)^{-1}\big)\lesssim \big\|\mathcal{O}\big\|^2_{H^1}+\frac{1}{m_Y^2}\big\|\chi(0)\big\|^2_{L^2}\]
    \[\big\|\alpha_Y\big\|^2_{L^2}\big((2m_Y)^{-1}\big)\lesssim \big\|\mathcal{O}\big\|^2_{L^2}+\frac{1}{m_Y^2}\big\|\nabla\mathcal{O}\big\|^2_{L^2}+\frac{1}{m_Y^4}\big\|\chi(0)\big\|^2_{L^2}\]
    \[\big\|\chi_Y\big\|^2_{L^2}\big((2m_Y)^{-1}\big)\lesssim \big\|\chi(0)\big\|^2_{L^2}+\frac{1}{m_Y^2}\big\|\mathcal{O}\big\|^2_{H^1}\]
    Combining all these inequalities we obtain the desired estimate.
\end{proof}

An immediate consequence of our previous estimates is the uniqueness to the initial value problem (\ref{main equation})-(\ref{definition of chi}) with given smooth asymptotic initial data:

\begin{proposition}[Uniqueness of scattering states]\label{uniqueness of smooth solutions}
    Let $\big(\alpha,\chi\big)$ and $\big(\alpha',\chi'\big)$ be smooth solutions of the system (\ref{main equation})-(\ref{definition of chi}), with the same asymptotic initial data given by $\chi(0),\mathcal{O},h\in C^{\infty}(S^n).$ Then $\alpha\equiv\alpha'$ and $\chi\equiv\chi'.$
\end{proposition}
\begin{proof}
    We define $\alpha_0=\alpha-\alpha',\ \chi_0=\chi-\chi'.$ By linearity,  $\big(\alpha_0,\chi_0\big)$ solves the system with vanishing asymptotic initial data. In particular, $\alpha_0$ satisfies the expansion:
    \[\alpha_0(\tau)=O(\tau^2|\log\tau|^2),\ \partial_{\tau}\alpha_0(\tau)=O(\tau|\log\tau|^2),\]
    Thus, we can apply the estimate in Corollary \ref{forward estimate alpha Y} to get $\alpha_0(1)=\partial_{\tau}\alpha_0(1)=\chi_0(1)=0.$ The standard uniqueness result for hyperbolic equations implies that $\alpha_0$ vanishes identically.
\end{proof}

We are now in the position to prove the main result of this section, estimate (\ref{forward estimate intro}). In particular, we prove that at $\tau=1$ the solution gains spatial regularity compared to the asymptotic initial data. To capture this from our previous estimates, we first need to introduce the Littlewood-Paley frequency decomposition, according to Appendix \ref{Littlewood-Paley decomposition}. Thus, for any $f\in L^2(S^n)$ we can write:
\[f=f_0+\sum_{l=1}^{\infty}f_l,\ f_0=P_{\leq1}f,\ f_l=P_{(2^{l-1},2^{l}]}f\]

We prove the main estimate of this section:

\begin{theorem}\label{improved regularity theorem} Let any $s\geq1.$ Let $\big(\alpha,\chi\big)$ be a smooth solution of (\ref{main equation})-(\ref{definition of chi}) with asymptotic initial data given by $\big(\chi(0),\mathcal{O},h\big)$. We define:
\[\mathfrak{h}:=h-(\log\nabla)\mathcal{O}:=h-\sum_{l=1}^{\infty}l\log2\cdot\mathcal{O}_l\]
Then, the solution $\big(\alpha,\chi\big)$ satisfies the following estimates:
\[\bigg(\big\|\alpha\big\|_{H^{s+1/2}}+\big\|\partial_{\tau}\alpha\big\|_{{H^{s-1/2}}}\bigg)\bigg|_{\tau=1}\lesssim \big\|\mathfrak{h}\big\|_{{H^{s}}}+\big\|\mathcal{O}\big\|_{{H^{s}}}+\big\|\chi(0)\big\|_{H^{s-1/2}}\]
\[\big\|\chi(1)\big\|_{H^{s+1/2}}\lesssim \big\|\mathfrak{h}\big\|_{{H^{s}}}+\big\|\mathcal{O}\big\|_{{H^{s}}}+\big\|\chi(0)\big\|_{H^{s+1/2}}\]
\end{theorem}
\begin{proof}
    We notice that we can write $\alpha=\sum_{l=0}^{\infty}\alpha_l,\ \chi=\sum_{l=0}^{\infty}\chi_l,$ and each $\big(\alpha_l,\chi_l\big)$ is the solution of (\ref{main equation})-(\ref{definition of chi}) with asymptotic initial data given by $\big(\chi_l(0),\mathcal{O}_l,h_l\big).$  We consider $\big((\alpha_J)_l,(\chi_J)_l\big)$ to be the solution of (\ref{main equation})-(\ref{definition of chi}) with asymptotic initial data given by $\big(\frac{1}{2}\chi_l(0),0,\mathfrak{h}_l\big)$ and $\big((\alpha_Y)_l,(\chi_Y)_l\big)$ to be the solution of (\ref{main equation})-(\ref{definition of chi}) with asymptotic initial data given by $\big(\frac{1}{2}\chi_l(0),\mathcal{O}_l,l\log2\cdot\mathcal{O}_l\big)$. Since $h_l=\mathfrak{h}_l+l\log2\cdot\mathcal{O}_l,$ the uniqueness result for smooth solutions proved above implies that $\alpha_l=(\alpha_J)_l+(\alpha_Y)_l$ and $\chi_l=(\chi_J)_l+(\chi_Y)_l$. 
    
    The first estimate follows by applying the estimates in Corollary \ref{forward estimate alpha J} and Corollary \ref{forward estimate alpha Y} for each component, with $m_Y=m_J=2^l.$ We get that for all $l\geq0:$
    \[\bigg(\big\|\alpha_l\big\|^2_{H^1}+\big\|\partial_{\tau}\alpha_l\big\|^2_{L^2}\bigg)\bigg|_{\tau=1}\lesssim 2^l\Big(\big\|\mathcal{O}_l\big\|^2_{L^2}+\big\|\mathfrak{h}_l\big\|^2_{L^2}\Big)+\frac{1}{2^l}\Big(\big\|\nabla\mathcal{O}_l\big\|^2_{L^2}+\big\|\nabla\mathfrak{h}_l\big\|^2_{L^2}\Big)+\big\|\chi_l(0)\big\|^2_{L^2}\]
    We multiply this relation by $2^{2ls-l}$ and sum for all $l\geq0.$ Using the definition of Sobolev spaces in Appendix \ref{Littlewood-Paley decomposition}, we obtain the first estimate.

    To prove the second estimate, we notice that for any fixed $l\geq0$ we have:
    \[\big\|\chi_l(1)\big\|_{{H^{s+1/2}}}^2\lesssim\big\|\chi_l(0)\big\|^2_{{H^{s+1/2}}}+\int_0^{2^{-l-1}}\tau^2\big\|\alpha_l\big\|_{H^{s+1/2}}^2d\tau+\int^1_{2^{-l-1}}\tau^2\big\|\alpha_l\big\|_{H^{s+1/2}}^2d\tau\lesssim\]\[\lesssim\big\|\chi_l(0)\big\|^2_{{H^{s+1/2}}}+\bigg(\big\|\mathfrak{h}_l\big\|^2_{{H^{s+1/2}}}+\big\|\mathcal{O}_l\big\|^2_{{H^{s+1/2}}}+\big\|\chi_l(0)\big\|^2_{H^{s-1/2}}\bigg)\cdot\int_0^{2^{-l-1}}\tau^2|\log\tau|^2d\tau+\]\[+\bigg(\big\|\mathfrak{h}_l\big\|^2_{{H^{s}}}+\big\|\mathcal{O}_l\big\|^2_{{H^{s}}}+\big\|\chi_l(0)\big\|^2_{H^{s-1/2}}\bigg)\cdot\int^1_{2^{-l-1}}\tau d\tau\lesssim\big\|\mathfrak{h}_l\big\|^2_{{H^{s}}}+\big\|\mathcal{O}_l\big\|^2_{{H^{s}}}+\big\|\chi_l(0)\big\|^2_{H^{s+1/2}}\]
    where we used Proposition \ref{low frequency alpha J proposition}, Proposition \ref{low frequency alpha Y proposition}, and Proposition \ref{high frequency forward estimate}. We sum for all $l\geq0$ to complete the proof of the second estimate.
\end{proof} 

\subsection{Existence and Uniqueness of Scattering States}\label{existence to ivp subsection}

In this section we prove the existence and uniqueness of solutions of (\ref{main equation})-(\ref{definition of chi}) with smooth asymptotic initial data in Theorem \ref{existence to the initial value problem theorem}. Together with the estimates proved in Theorem \ref{improved regularity theorem}, this completes the proof of the first statement in Theorem \ref{scattering theorem for main equation intro}, establishing the \textit{existence and uniqueness of scattering states}.

\begin{theorem}[Existence and uniqueness of scattering states]\label{existence to the initial value problem theorem}
    For any asymptotic initial data $\chi(0),\mathcal{O},h\in C^{\infty}(S^n)$, there exists a unique smooth solution $\big(\alpha,\chi\big)$ of (\ref{main equation})-(\ref{definition of chi}) with this initial data.
\end{theorem}
\begin{proof}
    We point out that we already proved uniqueness in Proposition \ref{uniqueness of smooth solutions}. We first prove the existence result for each component in the frequency decomposition using an iteration argument. For any $l,$ we introduce the time variable $t=2^l\tau$. The main equation can be written as:
    \[\partial_{t}^2\alpha_l+\frac{1}{t}\partial_{t}\alpha_l=-4q'\frac{\alpha_l}{2^{2l}}+\frac{4}{(t^2/2^{2l}+1)^2}\cdot\Delta\frac{\alpha_l}{2^{2l}}+\frac{f_1'(t)}{2^{2l}}t\partial_{t}\alpha_l+\frac{f_2'(t)}{2^{4l}}t^2\alpha_l+\frac{f_3'(t)}{2^{2l}}\cdot\bigg(\chi_l(0)+\frac{1}{2^{2l}}\int_0^t2t'\alpha_l(t')dt'\bigg)\]
    where $f'_1,f'_2,f'_3$ are bounded functions of $t.$ Additionally, we have that $\chi_l$ satisfies the equation:
    \begin{equation}\label{equation for chi l}
        \chi_l(t)=\chi_l(0)+\frac{1}{2^{2l}}\int_0^t2t'\alpha_l(t')dt'
    \end{equation}
    We can write the equation for $\alpha_l$ as:
    \begin{equation}\label{equation alpha l like on ODE and inhom term}
        \partial_{t}^2\alpha_l+\frac{1}{t}\partial_{t}\alpha_l=F\bigg(\alpha_l,\partial_{t}\alpha_l,\Delta\alpha_l,\int_0^tt'\alpha_l(t')dt',\chi_l(0)\bigg)
    \end{equation}
    We prove local existence for this equation using the following iteration scheme:
    \[\alpha_l^0=\mathcal{O}_l\log t+\mathfrak{h}_l\]
    \begin{align}\label{iteration scheme}
        \alpha_l^{n+1}=\mathcal{O}_l\log t+\mathfrak{h}_l&+\log t\int_0^tt'F\bigg(\alpha_l^n,\partial_{t}\alpha_l^n,\Delta\alpha_l^n,\int_0^{t'}t''\alpha_l^n(t'')dt'',\chi_l(0)\bigg)dt'-\\ \nonumber &-\int_0^tt'\log t'F\bigg(\alpha_l^n,\partial_{t}\alpha_l^n,\Delta\alpha_l^n,\int_0^{t'}t''\alpha_l^n(t'')dt'',\chi_l(0)\bigg)dt'
    \end{align}
    We define the renormalized sequence:
    \[\widetilde{\alpha_l^{n}}:=\alpha_l^{n}-\mathcal{O}_l\log t-\mathfrak{h}_l\]
    Note that $\widetilde{\alpha_l^{n}}\in C^1_t([0,\infty))C^{\infty}(S^n)$. For any $s\geq1$ we have the bound for all $t\leq 2^l$:
    \[\bigg\|F\bigg(\alpha_l^n,\partial_{t}\alpha_l^n,\Delta\alpha_l^n,\int_0^{t'}t''\alpha_l^n(t'')dt'',\chi_l(0)\bigg)-F\bigg(\alpha_l^{n-1},\partial_{t}\alpha_l^{n-1},\Delta\alpha_l^{n-1},\int_0^{t'}t''\alpha_l^{n-1}(t'')dt'',\chi_l(0)\bigg)\bigg\|_{H^s}\lesssim\]\[\lesssim\sup_{t'\in[0,t]}\big\|\widetilde{\alpha_l^{n}}-\widetilde{\alpha_l^{n-1}}\big\|_{H^s}+t\big\|\partial_t\widetilde{\alpha_l^{n}}-\partial_t\widetilde{\alpha_l^{n-1}}\big\|_{H^s}\]
    Using this bound, we have for $t\in[0,1]$ and any $s\geq1:$
    \[\big\|\widetilde{\alpha_l^{n+1}}-\widetilde{\alpha_l^{n}}\big\|_{H^s}(t)+\big\|\partial_t\widetilde{\alpha_l^{n+1}}-\partial_t\widetilde{\alpha_l^{n}}\big\|_{H^s}(t)\lesssim \big(t^2|\log(t)|+t\big)\bigg[\sup_{t'\in[0,t]}\big\|\widetilde{\alpha_l^{n}}-\widetilde{\alpha_l^{n-1}}\big\|_{H^s}+\sup_{t'\in[0,t]}\big\|\partial_t\widetilde{\alpha_l^{n}}-\partial_t\widetilde{\alpha_l^{n-1}}\big\|_{H^s}\bigg]\]
    The constant in the above inequality is independent of $l,t,$ and $n.$ Thus, there exists $\epsilon>0$ such that for $t\in[0,\epsilon]$, the sequence $\big\{(\widetilde{\alpha_l^{n}},\partial_t\widetilde{\alpha_l^{n}})\big\}$ is a contraction in $H^s$, so there exists $\widetilde{\alpha_l}\in C^1_t\big([0,\epsilon];H^s(S^n)\big)$ such that $\widetilde{\alpha_l^{n}}\rightarrow\widetilde{\alpha_l}.$ We define $\alpha_l=\widetilde{\alpha_l}+\mathcal{O}_l\log t+\mathfrak{h}_l.$ Taking the limit in (\ref{iteration scheme}), we obtain that $\alpha_l$ satisfies the equation for $t\in[0,\epsilon]$:
    \begin{align}\label{duhamel type equation alpha l}
        \alpha_l=\mathcal{O}_l\log t+\mathfrak{h}_l&+\log t\int_0^tt'F\bigg(\alpha_l,\partial_{t}\alpha_l,\Delta\alpha_l,\int_0^{t'}t''\alpha_l(t'')dt'',\chi_l(0)\bigg)dt'-\\ \nonumber &-\int_0^tt'\log t'F\bigg(\alpha_l,\partial_{t}\alpha_l,\Delta\alpha_l,\int_0^{t'}t''\alpha_l(t'')dt'',\chi_l(0)\bigg)dt'
    \end{align}
    This implies that $\alpha_l$ satisfies (\ref{equation alpha l like on ODE and inhom term}) for $t\in[0,\epsilon].$ Since (\ref{equation alpha l like on ODE and inhom term}) is a linear hyperbolic equation, we obtain that $\alpha_l$ extends as a $C^1_t\big((0,\infty);H^s(S^n)\big)$ solution. This holds for any $s\geq1,$ so $\alpha_l$ is a smooth solution on $(0,\infty)\times S^n.$ A direct computation shows that the RHS of (\ref{duhamel type equation alpha l}) also defines a smooth solution on $(0,\infty)\times S^n.$ By uniqueness, we obtain that $\alpha_l$ satisfies (\ref{duhamel type equation alpha l})  for $t\in[0,\infty).$ This also implies that $\widetilde{\alpha_l}\in C^1_t([0,\infty))C^{\infty}(S^n)$. Moreover, (\ref{equation for chi l}) implies that $\chi_l\in C^1_t([0,\infty))C^{\infty}(S^n).$ We notice that for $t\in(0,2^l]$ we have the bound:
    \[\bigg\|F\bigg(\alpha_l,\partial_{t}\alpha_l,\Delta\alpha_l, \int_0^{t}t'\alpha_l(t'),\chi_l(0)\bigg)\bigg\|_{H^s}\lesssim\big(1+|\log t|\big)\cdot\|\mathcal{O}_l\|_{H^s}+\|\mathfrak{h}_l\|_{H^s}+\|\chi_l(0)\|_{H^s}+\]\[+\sup_{t'\in[0,t]}\big(\|\widetilde{\alpha_l}\|_{H^s}+\|t'\partial_t\widetilde{\alpha_l}\|_{H^s}\big)\]
    Using this bound in (\ref{duhamel type equation alpha l}), we obtain that $\alpha_l$ satisfies the expansions: 
    \[\alpha_l(t)=\mathcal{O}_l\log(t)+\mathfrak{h}_l+O\big({t}^2|\log({t})|^2\big),\ \partial_{t}\alpha_l({t})=\frac{\mathcal{O}}{t}+O\big({t}|\log({t})|^2\big)\text{ in }C^{\infty}(S^n)\]
    In particular, this proves that $\big(\alpha_l,\chi_l\big)$ defines a smooth solution of (\ref{main equation})-(\ref{definition of chi}) with asymptotic initial data $\big(\chi_l(0),\mathcal{O}_l,h_l\big).$ Thus, we proved so far the existence of solutions with asymptotic initial data supported in a frequency strip.
    
    In order to sum up all components in the frequency decomposition, we need to make the above estimates quantitative. According to the first part of the proof, we can now use the estimates from the previous section for solutions supported in a frequency strip. By Proposition \ref{low frequency alpha J proposition} and Proposition \ref{low frequency alpha Y proposition}, we obtain that for any $s\geq1$ and any $t\in[0,1/2]:$
    \[\big\|\alpha_l\big\|_{H^s}+\big\|t\partial_t\alpha_l\big\|_{H^s}\lesssim|\log(t)|\cdot\big(\|\mathcal{O}_l\|_{H^s}+\|\mathfrak{h}_l\|_{H^s}+\|\chi_l(0)\|_{H^s}\big)\]
    As a result, we have for any $s\geq1$ and any $t\in[0,1/2]:$
    \[\bigg\|F\bigg(\alpha_l,\partial_{t}\alpha_l,\Delta\alpha_l, \int_0^{t}t'\alpha_l(t')dt',\chi_l(0)\bigg)\bigg\|_{H^s}\lesssim|\log t|\cdot\big(\|\mathcal{O}_l\|_{H^s}+\|\mathfrak{h}_l\|_{H^s}+\|\chi_l(0)\|_{H^s}\big)\]
    Using the estimates in Proposition \ref{high frequency forward estimate} (as in the proofs of Corollary \ref{forward estimate alpha J} and Corollary \ref{forward estimate alpha Y}), we have for any $s\geq1$ and any $t\in[1/2,2^l]:$
    \[\big\|\alpha_l\big\|_{H^s}+\big\|t\partial_t\alpha_l\big\|_{H^s}\lesssim\sqrt{t}\cdot\big(\|\mathcal{O}_l\|_{H^s}+\|\mathfrak{h}_l\|_{H^s}+\|\chi_l(0)\|_{H^s}\big)\]
    This implies that for any $s\geq1$ and any $t\in[1/2,2^l]:$
    \[\bigg\|F\bigg(\alpha_l,\partial_{t}\alpha_l,\Delta\alpha_l, \int_0^{t}t'\alpha_l(t')dt',\chi_l(0)\bigg)\bigg\|_{H^s}\lesssim2^{l/2}\cdot\big(\|\mathcal{O}_l\|_{H^s}+\|\mathfrak{h}_l\|_{H^s}+\|\chi_l(0)\|_{H^s}\big)\]
    Using these bounds on $F$ in (\ref{duhamel type equation alpha l}) and (\ref{equation for chi l}), we get for any $s\geq1$ and any $t\in[0,2^l]:$
    \[\big\|\alpha_l-\mathcal{O}_l\log t-\mathfrak{h}_l\big\|_{H^s}\lesssim t^2\big(1+|\log t|^2\big)2^{l/2}\cdot\big(\|\mathcal{O}_l\|_{H^s}+\|\mathfrak{h}_l\|_{H^s}+\|\chi_l(0)\|_{H^s}\big)\]
    \[\big\|\partial_t\alpha_l-\mathcal{O}_l/t\big\|_{H^s}\lesssim t\big(1+|\log t|^2\big)2^{l/2}\cdot\big(\|\mathcal{O}_l\|_{H^s}+\|\mathfrak{h}_l\|_{H^s}+\|\chi_l(0)\|_{H^s}\big)\]
    \[\big\|\chi_l-\chi_l(0)\big\|_{H^s}\lesssim t^2\big(1+|\log t|^2\big)\cdot\big(\|\mathcal{O}_l\|_{H^s}+\|\mathfrak{h}_l\|_{H^s}+\|\chi_l(0)\|_{H^s}\big)\]
    \[\big\|\partial_t\chi_l\big\|_{H^s}\lesssim t\big(1+|\log t|^2\big)\cdot\big(\|\mathcal{O}_l\|_{H^s}+\|\mathfrak{h}_l\|_{H^s}+\|\chi_l(0)\|_{H^s}\big)\]
    In terms of the coordinate $\tau,$ we proved that for any $s\geq1$ and any $\tau\in[0,1]:$
    \[\big\|\alpha_l-\mathcal{O}_l\log(2^l\tau)-\mathfrak{h}_l\big\|_{H^s}\lesssim \tau^2\big(1+|\log\tau|^2\big)\cdot\big(\|\mathcal{O}_l\|_{H^{s+3}}+\|\mathfrak{h}_l\|_{H^{s+3}}+\|\chi_l(0)\|_{H^{s+3}}\big)\]
    \[\bigg\|\partial_{\tau}\alpha_l-\frac{\mathcal{O}_l}{\tau}\bigg\|_{H^s}\lesssim \tau\big(1+|\log\tau|^2\big)\cdot\big(\|\mathcal{O}_l\|_{H^{s+3}}+\|\mathfrak{h}_l\|_{H^{s+3}}+\|\chi_l(0)\|_{H^{s+3}}\big)\]
    \[\big\|\chi_l-\chi_l(0)\big\|_{H^s}\lesssim \tau^2\big(1+|\log\tau|^2\big)\cdot\big(\|\mathcal{O}_l\|_{H^{s+3}}+\|\mathfrak{h}_l\|_{H^{s+3}}+\|\chi_l(0)\|_{H^{s+3}}\big)\]
    \[\big\|\partial_{\tau}\chi_l\big\|_{H^s}\lesssim \tau\big(1+|\log\tau|^2\big)\cdot\big(\|\mathcal{O}_l\|_{H^{s+3}}+\|\mathfrak{h}_l\|_{H^{s+3}}+\|\chi_l(0)\|_{H^{s+3}}\big)\]
    The constants in the above inequalities are independent of $l,$ so we can define $\alpha=\sum_{l=0}^{\infty}\alpha_l,\ \chi=\sum_{l=0}^{\infty}\chi_l$ and we proved that for any $s\geq1$ and any $\tau\in[0,1]:$
    \[\big\|\alpha({\tau})-\mathcal{O}\log({\tau})-h\big\|_{H^s}\lesssim \tau^2\big(1+|\log\tau|^2\big)\cdot\big(\|\mathcal{O}\|_{H^{s+3}}+\|\mathfrak{h}\|_{H^{s+3}}+\|\chi(0)\|_{H^{s+3}}\big)\]
    \[\bigg\|\partial_{\tau}\alpha-\frac{\mathcal{O}}{\tau}\bigg\|_{H^s}\lesssim \tau\big(1+|\log\tau|^2\big)\cdot\big(\|\mathcal{O}\|_{H^{s+3}}+\|\mathfrak{h}\|_{H^{s+3}}+\|\chi(0)\|_{H^{s+3}}\big)\]
    together with similar estimates for $\chi.$ We get that $\alpha-\mathcal{O}\log({\tau})-h\in C^1_{\tau}([0,1])C^{\infty}(S),\ \chi\in C^1_{\tau}([0,1])C^{\infty}(S)$ and that the expansions (\ref{main expansion}) hold. Moreover, $(\alpha,\chi)$ solves the linear system (\ref{main equation})-(\ref{definition of chi}), so by propagation of regularity with initial data at $\tau=1/2$ we also have $\alpha,\chi\in C^{\infty}\big((0,\infty)\times S\big).$ In conclusion, $\big(\alpha,\chi\big)$ defines a smooth solution of (\ref{main equation})-(\ref{definition of chi}) with asymptotic initial data $\big(\chi(0),\mathcal{O},h\big).$
\end{proof}

\subsection{Estimates from a Finite Time Hypersurface to $\mathcal{I}^+$}\label{backward direction subsection}

In this section we assume the existence of smooth solutions of (\ref{main equation})-(\ref{definition of chi}) which satisfy the standard asymptotic expansion at $\tau=0,$ and prove estimates on the asymptotic data in terms of the initial data at time $\tau=1$. In the context of (\ref{wave equation de sitter intro}), these correspond to estimates on the asymptotic data at $\mathcal{I}^-$ in terms of $\frac{n}{2}$ time derivatives of $\Tilde{\phi}$ at a finite time hypersurface. Since (\ref{wave equation de sitter intro}) is time reversible, we also obtain estimates on the asymptotic data at $\mathcal{I}^+$ in terms of $\frac{n}{2}$ time derivatives of $\Tilde{\phi}$ at a finite time hypersurface. The main result in Theorem \ref{main theorem backwards direction} proves estimate (\ref{backward estimate intro}) of Theorem \ref{scattering theorem for main equation intro}.

As in Section \ref{forward direction subsection}, we treat separately the low frequency regime and the high frequency regime. In this case, the transition time is determined by the frequency strip that our solution is localized in. We begin with the following high frequency regime estimate, which is similar to Proposition \ref{high frequency forward estimate}:

\begin{proposition}\label{high frequency backward estimate}
    Let any $s\geq1.$ For any $\tau_0\in(0,1],$ we set $\alpha_H=P_{\geq 1/4\tau_0}\alpha.$ Then for any $\tau\in[\tau_0,1]$:
    \[\frac{1}{\tau}\big\|\alpha_H\big\|^2_{H^s}+\tau\big\|\nabla\alpha_H\big\|^2_{H^s}+\tau\big\|\partial_{\tau}\alpha_H\big\|^2_{H^s}\lesssim \bigg(\big\|\alpha_H\big\|^2_{H^{s+1}}+\big\|\partial_{\tau}\alpha_H\big\|^2_{H^s}+\big\|\chi_H\big\|^2_{H^s}\bigg)\bigg|_{\tau=1}\]
    where the implicit constant is independent of $\tau_0$.
\end{proposition}
\begin{proof}
    We set $m=\frac{1}{\tau_0}$, and we introduce the new time variable $t=m\tau\in[1,m]$. As before, equation (\ref{main equation}) can be written as:
    \[\partial_{t}^2\big(\alpha_H\sqrt{t}\big)+\frac{1}{4t^2}\alpha_H\sqrt{t}+\frac{4q'}{m^2}\alpha_H\sqrt{t}-\frac{4}{(t^2/m^2+1)^2}\cdot\Delta\frac{\alpha_H\sqrt{t}}{m^2}=\frac{f_1'(t)}{m^2}t\partial_{t}\big(\alpha_H\sqrt{t}\big)+\frac{f_2'(t)}{m^2}\alpha_H\sqrt{t}+\frac{f_3'(t)}{m^2}\chi_H\sqrt{t}\]
    where $f'_1,f'_2,f'_3$ are bounded functions of $t.$ We also have from (\ref{definition of chi}):
     \[\chi_H(t)=\chi_H(m)-\frac{1}{m^2}\int_t^m2t'\alpha_H(t')dt'\]
    which implies the bound for $t\in[1,m]$:
    \[\big\|\chi_H\big\|^2_{L^2}\lesssim\big\|\chi(m)\big\|^2_{L^2}+\frac{1}{m^2}\int_{t}^m\big\|\alpha_H\sqrt{t'}\big\|^2_{L^2}dt'\]
    In the following inequalities, the implicit constant will be independent of $m$ and $t$. We have the standard energy estimate:
    \[\big\|\partial_{t}\big(\alpha_H\sqrt{t}\big)\big\|^2_{L^2}+\frac{1}{t^2}\big\|\alpha_H\sqrt{t}\big\|^2_{L^2}+\frac{1}{m^2}\big\|\nabla\alpha_H\sqrt{t}\big\|^2_{L^2}\lesssim\]\[\lesssim\bigg(t\big\|\partial_{t}\alpha_H\big\|^2_{L^2}+\frac{1}{t}\big\|\alpha_H\big\|^2_{L^2}+\frac{1}{m}\big\|\nabla\alpha_H\big\|^2_{L^2}\bigg)\bigg|_{t=m}+\int_{t}^m\frac{1}{t'^3}\big\|\alpha_H\sqrt{t'}\big\|^2_{L^2}dt'+\]\[+\frac{1}{m^4}\int_{t}^mt'\big\|\nabla\alpha_H\sqrt{t'}\big\|^2_{L^2}dt'+\int_{t}^m\frac{t'}{m^2}\big\|\partial_{t}\big(\alpha_H\sqrt{t'}\big)\big\|^2_{L^2}+\int_{t}^m\int_S\frac{\sqrt{t'}}{m^2}\big(|\chi_H|+|\alpha_H|\big)\cdot\big|\partial_{t}\big(\alpha_H\sqrt{t'}\big)\big|\]
    We remark that the bulk terms have an unfavorable sign, which forces us to restrict to high frequencies. We notice that we can bound:
    \[\int_{t}^m\int_S\frac{\sqrt{t'}}{m^2}\big(|\chi_H|+|\alpha_H|\big)\cdot\big|\partial_{t}\big(\alpha_H\sqrt{t'}\big)\big|\lesssim\int_{t}^m\frac{t'}{m^2}\big\|\partial_{t}\big(\alpha_H\sqrt{t'}\big)\big\|^2_{L^2}+\frac{1}{m^2}\int_{t}^m\big\|\alpha_H\big\|^2_{L^2}+\big\|\chi_H\big\|^2_{L^2}\lesssim\]\[\lesssim\frac{1}{m}\big\|\chi_H(m)\big\|^2_{L^2}+\int_{t}^m\frac{t'}{m^2}\big\|\partial_{t}\big(\alpha_H\sqrt{t'}\big)\big\|^2_{L^2}dt'+\frac{1}{m^2}\int_{t}^m\big\|\alpha_H\big\|^2_{L^2}dt'\]
    Combining the previous two inequalities we have for all $t\in[1,m]$:
    \[\big\|\partial_{t}\big(\alpha_H\sqrt{t}\big)\big\|^2_{L^2}+\frac{1}{t}\big\|\alpha_H\big\|^2_{L^2}+\frac{t}{m^2}\big\|\nabla\alpha_H\big\|^2_{L^2}\lesssim\bigg(t\big\|\partial_{t}\alpha_H\big\|^2_{L^2}+\frac{1}{m}\big\|\alpha_H\big\|^2_{H^1}+\frac{1}{m}\big\|\chi_H\big\|^2_{L^2}\bigg)\bigg|_{t=m}+\]\[+\int_{t}^m\frac{1}{t'^2}\big\|\alpha_H\big\|^2_{L^2}dt'+\frac{1}{m^4}\int_{t}^mt'^2\big\|\nabla\alpha_H\big\|^2_{L^2}dt'+\int_{t}^m\frac{t'}{m^2}\big\|\partial_{t}\big(\alpha_H\sqrt{t'}\big)\big\|^2_{L^2}\]
    To deal with the first bulk term, we use the Poincar\'{e} inequality:
    \[\big\|\alpha_H\big\|_{L^2}\lesssim\frac{1}{m}\big\|\nabla\alpha_H\big\|_{L^2},\]
    which follows from the definition of $\alpha_H.$
    Since $1\leq t\leq m,$ we can apply Gronwall to obtain:
    \[\big\|\partial_{t}\big(\alpha_H\sqrt{t}\big)\big\|^2_{L^2}+\frac{1}{t}\big\|\alpha_H\big\|^2_{L^2}+\frac{t}{m^2}\big\|\nabla\alpha_H\big\|^2_{L^2}\lesssim\bigg(t\big\|\partial_{t}\alpha_H\big\|^2_{L^2}+\frac{1}{m}\big\|\alpha_H\big\|^2_{H^1}+\frac{1}{m}\big\|\chi_H\big\|^2_{L^2}\bigg)\bigg|_{t=m}\]
    Finally, we can rewrite this estimate to get that for any $\tau\in[\tau_0,1]$:
    \begin{equation}\label{high frequency backward estimate L2}
        \frac{1}{\tau}\big\|\alpha_H\big\|^2_{L^2}+\tau\big\|\nabla\alpha_H\big\|^2_{L^2}+\tau\big\|\partial_{\tau}\alpha_H\big\|^2_{L^2}\lesssim\bigg(\big\|\alpha_H\big\|^2_{H^1}+\big\|\partial_{\tau}\alpha_H\big\|^2_{L^2}+\big\|\chi_H\big\|^2_{L^2}\bigg)\bigg|_{\tau=1}
    \end{equation}
    If $\tau_0<1/2$, we denote by $l_0\geq1$ the smallest integer such that $2^{l_0+1}>m.$ We set $\alpha^H_0=P_{[m/4,2^{l_0-1}]}\alpha$ and $\chi^H_0=P_{[m/4,2^{l_0-1}]}\chi$. By Appendix \ref{Littlewood-Paley decomposition}, we have the decomposition:
    \[\alpha_H=\alpha^H_0+\sum_{l=l_0}^{\infty}\alpha_l,\ \chi_H=\chi^H_0+\sum_{l=l_0}^{\infty}\chi_l\]
    Moreover, (\ref{high frequency backward estimate L2}) applies for $(\alpha^H_0,\chi^H_0)$ and each $(\alpha_l,\chi_l)$ with $l\geq l_0,$ where the implicit constant is independent of $l$ and $\tau_0.$ We multiply these inequalities by $2^{2(l_0-1)s},$ respectively $2^{2ls}$ with $l\geq l_0$, in order to obtain the corresponding estimates for the $H^s$ norms. Summing all the components in the above decomposition we obtain the desired inequality. 
    
    If $\tau_0\geq1/2$ we have instead the decomposition:
    \[\alpha_H=P_{[m/4,1]}\alpha+\sum_{l=1}^{\infty}\alpha_l,\ \chi_H=P_{[m/4,1]}\chi+\sum_{l=1}^{\infty}\chi_l\]
    and we repeat the same argument as above to complete the proof.
\end{proof}

In the low frequency regime we use different multipliers from Section \ref{forward direction subsection}, in order to obtain favorable bulk terms. As pointed out in the introduction, we anticipate the need to renormalize $h$ to $\mathfrak{h}$ at the level of $\alpha,$ without explicit reference to the asymptotics of the solution. This can be seen in the third term controlled in the energy estimate below.

We prove the following low frequency estimate:

\begin{proposition}\label{low frequency backward estimate}
    Let any $s\geq1.$ For any $\tau_0\in(0,1]$ and for all $\tau\in[\tau_0,2^{-l-1}]\subset(0,1]$ we have:
    \[\|\tau\partial_{\tau}\alpha_l\|^2_{H^s}+\|\tau\alpha_l\|^2_{H^{s+1}}+\|\alpha_l-\log(2^l\tau)\cdot\tau\partial_{\tau}\alpha_l\|^2_{H^s}\lesssim\bigg(\|\alpha_l\|^2_{H^s}+\|\tau\nabla\alpha_l\|^2_{H^s}+\|\tau\partial_{\tau}\alpha_l\|^2_{H^s}+\big\|\tau\chi_l\big\|^2_{H^s}\bigg)\bigg|_{\tau=2^{-l-1}}\]
    where the implicit constant is independent of $\tau_0$ and $l$.
\end{proposition}
\begin{proof}
    We set $m=2^l$ and we introduce the new time variable $t=m\tau\in[m\tau_0,1/2]$. Equation (\ref{main equation}) can be written as:
    \[\partial_{t}\big(t\partial_t\alpha\big)+4q't\frac{\alpha}{m^2}-\frac{4t}{(t^2/m^2+1)^2}\cdot\Delta\frac{\alpha}{m^2}=\frac{f_1'(t)}{m^2}t^2\partial_{t}\alpha+\frac{f_2'(t)}{m^4}t^3\alpha+\frac{f_3'(t)}{m^2}t\chi\]
    where $f'_1,f'_2,f'_3$ are bounded functions of $t.$ We also have from (\ref{definition of chi}):
    \[\chi(t)=\chi(1/2)-\frac{1}{m^2}\int^{1/2}_t2t'\alpha(t')dt'\]
    which implies the bound for $t\in[m\tau_0,1/2]$:
    \begin{equation}\label{preliminary bound for chi in low freq backwards direction}
        \big\|\chi\big\|^2_{L^2}\lesssim\big\|\chi(1/2)\big\|^2_{L^2}+\frac{1}{m^4}\int^{1/2}_t\big\|t'\alpha\big\|^2_{L^2}dt'
    \end{equation}
    We multiply the equation by $t\partial_t\alpha$ and integrate by parts. Since for $t\in[m\tau_0,1/2]$ we have:
    \[\frac{d}{dt}\bigg(\frac{4q't^2}{m^2}\bigg)\geq0,\ \frac{d}{dt}\bigg(\frac{4t^2/m^2}{(t^2/m^2+1)^2}\bigg)\geq0,\]
    we obtain that the bulk terms will have favorable signs in the energy estimate:
    \[\big\|t\partial_{t}\alpha\big\|^2_{L^2}+\frac{t^2}{m^2}\big\|\alpha\big\|^2_{L^2}+\int^{1/2}_t\frac{t'}{m^2}\big\|\alpha\big\|^2_{L^2}dt'+\frac{t^2}{m^2}\big\|\nabla\alpha\big\|^2_{L^2}+\int^{1/2}_t\frac{t'}{m^2}\big\|\nabla\alpha\big\|^2_{L^2}dt'\lesssim\]\[\lesssim\bigg(\big\|t\partial_{t}\alpha\big\|^2_{L^2}+\frac{1}{m^2}\big\|\alpha\big\|^2_{H^1}\bigg)\bigg|_{t=\frac{1}{2}}+\int^{1/2}_t\frac{t'}{m^2}\big\|t'\partial_{t}\alpha\big\|^2_{L^2}+\int^{1/2}_t\int_S\bigg(\frac{t'}{m^2}|\chi|+\frac{t'^3}{m^4}|\alpha|\bigg)\cdot\big|t'\partial_{t}\alpha\big|\]
    We remark that we can bound:
    \[\int^{1/2}_t\int_S\bigg(\frac{t'}{m^2}|\chi|+\frac{t'^3}{m^4}|\alpha|\bigg)\cdot\big|t'\partial_{t}\alpha\big|\lesssim\int^{1/2}_t\frac{t'}{m^2}\big\|t'\partial_{t}\alpha\big\|^2_{L^2}+\frac{1}{m^2}\int^{1/2}_t\frac{1}{m^2}\big\|t'\alpha\big\|^2_{L^2}+\big\|\chi\big\|^2_{L^2}\lesssim\]\[\lesssim\frac{1}{m^2}\big\|\chi(1/2)\big\|^2_{L^2}+\int^{1/2}_t\frac{t'}{m^2}\big\|t'\partial_{t}\alpha\big\|^2_{L^2}+\frac{1}{m^4}\int^{1/2}_t\big\|t'\alpha\big\|^2_{L^2}dt'\]
    Combining the previous two inequalities we have:
    \[\big\|t\partial_{t}\alpha\big\|^2_{L^2}+\frac{t^2}{m^2}\big\|\alpha\big\|^2_{H^1}\lesssim\bigg(\big\|t\partial_{t}\alpha\big\|^2_{L^2}+\frac{1}{m^2}\big\|\alpha\big\|^2_{H^1}+\frac{1}{m^2}\big\|\chi\big\|^2_{L^2}\bigg)\bigg|_{t=\frac{1}{2}}+\int^{1/2}_t\frac{t'}{m^2}\big\|t'\partial_{t}\alpha\big\|^2_{L^2}+\frac{1}{m^4}\big\|t'\alpha\big\|^2_{L^2}\]
    By Gronwall we obtain that for all $t\in[m\tau_0,1/2]$:
    \begin{equation}\label{first half of low frequency backwards estimate}
        \big\|t\partial_{t}\alpha\big\|^2_{L^2}+\frac{t^2}{m^2}\big\|\alpha\big\|^2_{H^1}\lesssim\bigg(\big\|t\partial_{t}\alpha\big\|^2_{L^2}+\frac{1}{m^2}\big\|\alpha\big\|^2_{H^1}+\frac{1}{m^2}\big\|\chi\big\|^2_{L^2}\bigg)\bigg|_{t=\frac{1}{2}}
    \end{equation}
    In terms of the $\tau$ variable we get that for all $\tau\in[\tau_0,2^{-l-1}]:$
    \[\|\tau\partial_{\tau}\alpha\|^2_{L^2}+\|\tau\alpha\|^2_{H^1}\lesssim\bigg(\|\tau\alpha\|^2_{H^1}+\|\tau\partial_{\tau}\alpha\|^2_{L^2}+\big\|\tau\chi\big\|^2_{L^2}\bigg)\bigg|_{\tau=2^{-l-1}}\]
    The implicit constant in this inequality is independent of $l$ and $\tau.$ We can apply this result for $(\alpha_l,\chi_l)$ and multiply by $2^{2sl}$ to obtain the corresponding inequality for $H^s$ norms:
    \[\|\tau\partial_{\tau}\alpha_l\|^2_{H^s}+\|\tau\alpha_l\|^2_{H^{s+1}}\lesssim\bigg(\|\tau\alpha_l\|^2_{H^{s+1}}+\|\tau\partial_{\tau}\alpha_l\|^2_{H^s}+\big\|\tau\chi_l\big\|^2_{H^s}\bigg)\bigg|_{\tau=2^{-l-1}}\]

    For the second part of the proof, we write (\ref{main equation}) as:
    \[\partial_{t}\big(\alpha_l-t\log t\partial_t\alpha_l\big)-4q't\log t\frac{\alpha_l}{m^2}+\frac{4t\log t}{(t^2/m^2+1)^2}\cdot\Delta\frac{\alpha_l}{m^2}=\frac{f_1'(t)}{m^2}t^2\log t\partial_{t}\alpha_l+\frac{f_2'(t)}{m^4}t^3\log t\alpha_l+\frac{f_3'(t)}{m^2}t\log t\chi_l\]
    where $f'_1,f'_2,f'_3$ are bounded functions of $t.$ It is convenient to introduce the notation:
    \[\overline{\alpha}_l=\alpha_l-t\log t\partial_t\alpha_l\]
    We multiply the equation by $\overline{\alpha}_l$ and we obtain:
    \[\big\|\overline{\alpha}_l\big\|_{L^2}^2(t)\lesssim\big\|\overline{\alpha}_l\big\|_{L^2}^2(1/2)+\int_t^{1/2}\int_St'|\log t'|\cdot|\overline{\alpha}_l|\cdot\bigg(|\alpha_l|+\frac{|\Delta\alpha_l|}{m^2}+\frac{|t'\partial_t\alpha_l|}{m^2}+\frac{|\chi_l|}{m^2}\bigg)dt'\lesssim\]\[\lesssim\big\|\overline{\alpha}_l\big\|_{L^2}^2(1/2)+\int_t^{1/2}t'|\log t'|\cdot\|\overline{\alpha}_l\|_{L^2}\cdot\bigg(\|\overline{\alpha}_l\|_{L^2}+\|t'\log t'\cdot\partial_t\alpha_l\|_{L^2}+\frac{\|\chi_l\|_{L^2}}{m^2}\bigg)dt'\lesssim\]\[\lesssim\big\|\overline{\alpha}_l\big\|_{L^2}^2(1/2)+\int_t^{1/2}t'|\log t'|\cdot\big\|\overline{\alpha}_l\big\|^2_{L^2}dt'+\int_t^{1/2}t'|\log t'|^3\cdot\big\|t'\partial_t\alpha_l\|^2_{L^2}+\frac{t'|\log t'|}{m^4}\cdot\big\|\chi_l\big\|_{L^2}^2dt'\]
    We now use (\ref{preliminary bound for chi in low freq backwards direction}) and (\ref{first half of low frequency backwards estimate}) to obtain:
    \[\big\|\overline{\alpha}_l\big\|_{L^2}^2(t)\lesssim\bigg(\big\|\overline{\alpha}_l\big\|_{L^2}^2+\big\|t\partial_{t}\alpha_l\big\|^2_{L^2}+\frac{1}{m^2}\big\|\alpha_l\big\|^2_{H^1}+\frac{1}{m^2}\big\|\chi_l\big\|^2_{L^2}\bigg)\bigg|_{t=\frac{1}{2}}+\int_t^{1/2}t'|\log t'|\cdot\big\|\overline{\alpha}_l\big\|^2_{L^2}dt'\]
    We apply Gronwall, and convert to the $\tau$ variable as before. Finally, we multiply the inequality by $2^{2sl}$ to obtain the desired inequality for $H^s$ norms.
\end{proof}

We combine the low frequency and high frequency estimates for our solution in order to obtain:

\begin{corollary}\label{combined backwards direction estimate}
    Let any $s\geq1.$ For any $\tau_0\in(0,1/2],$ we have:
    \[\sum_{l=0}^{\infty}\textbf{1}_{\{2^{l+1}\leq\tau_0^{-1}\}}\cdot\bigg(\|\tau\partial_{\tau}\alpha_l\|^2_{H^s}+\|\alpha_l-\log(2^l\tau)\cdot\tau\partial_{\tau}\alpha_l\|^2_{H^s}\bigg)\bigg|_{\tau=\tau_0}\lesssim\bigg(\big\|\alpha\big\|^2_{H^{s+1/2}}+\big\|\partial_{\tau}\alpha\big\|^2_{{H^{s-1/2}}}+\big\|\chi\big\|^2_{{H^{s-1/2}}}\bigg)\bigg|_{\tau=1}\]
    where the implicit constant is independent of $\tau_0$.
\end{corollary}
\begin{proof}
    For each $l\geq0$ such that $2^{l+1}\leq\tau_0^{-1},$ we apply Proposition \ref{low frequency backward estimate} to obtain:
    \[\bigg(\|\tau\partial_{\tau}\alpha_l\|^2_{H^s}+\|\alpha_l-\log(2^l\tau)\cdot\tau\partial_{\tau}\alpha_l\|^2_{H^s}\bigg)\bigg|_{\tau=\tau_0}\lesssim\bigg(\|\alpha_l\|^2_{H^s}+\|\tau\nabla\alpha_l\|^2_{H^s}+\|\tau\partial_{\tau}\alpha_l\|^2_{H^s}+\big\|\tau\chi_l\big\|^2_{H^s}\bigg)\bigg|_{\tau=2^{-l-1}}\]
    Next, we remark that on $[2^{-l-1},1]$ we have $(\alpha_l)_H=P_{\geq2^{l-1}}\alpha_l=\alpha_l,$ so we get by Proposition \ref{high frequency backward estimate}:
    \[\bigg(\|\alpha_l\|^2_{H^s}+\|\tau\nabla\alpha_l\|^2_{H^s}+\|\tau\partial_{\tau}\alpha_l\|^2_{H^s}+\big\|\tau\chi_l\big\|^2_{H^s}\bigg)\bigg|_{\tau=2^{-l-1}}\lesssim 2^{-l}\cdot\bigg(\big\|\alpha_l\big\|^2_{H^{s+1}}+\big\|\partial_{\tau}\alpha_l\big\|^2_{H^s}+\big\|\chi_l\big\|^2_{H^s}\bigg)\bigg|_{\tau=1}\]
    The constants are independent of $l$ and $\tau_0,$ so we can sum the above inequalities.
\end{proof}

Putting together the estimates proved so far, we complete the proof of (\ref{backward estimate intro}):

\begin{theorem}\label{main theorem backwards direction}
    Let $\big(\alpha,\chi\big)$ be a smooth solution of (\ref{main equation})-(\ref{definition of chi}) with asymptotic data given by $\chi(0),\mathcal{O},$ and $h$. Then we have the estimates:
    \[\big\|\mathfrak{h}\big\|_{{H^{s}}}+\big\|\mathcal{O}\big\|_{{H^{s}}}\lesssim \bigg(\big\|\alpha\big\|_{H^{s+1/2}}+\big\|\partial_{\tau}\alpha\big\|_{{H^{s-1/2}}}+\big\|\chi\big\|_{{H^{s-1/2}}}\bigg)\bigg|_{\tau=1}\]
    \[\big\|\chi(0)\big\|_{{H^{s+1/2}}}\lesssim \bigg(\big\|\alpha\big\|_{H^{s+1/2}}+\big\|\partial_{\tau}\alpha\big\|_{{H^{s-1/2}}}+\big\|\chi\big\|_{{H^{s+1/2}}}\bigg)\bigg|_{\tau=1}\]
\end{theorem}
\begin{proof}
    The asymptotic expansion of the solution implies that for each $l\geq0$:
    \[\lim_{\tau_0\rightarrow0}\textbf{1}_{\{2^{l+1}\leq\tau_0^{-1}\}}\cdot\bigg(\|\tau\partial_{\tau}\alpha_l\|^2_{H^s}+\|\alpha_l-\log(2^l\tau)\cdot\tau\partial_{\tau}\alpha_l\|^2_{H^s}\bigg)\bigg|_{\tau=\tau_0}=\big\|\mathcal{O}_l\big\|_{{H^{s}}}^2+\big\|\mathfrak{h}_l\big\|_{{H^{s}}}^2\]
    We use Corollary \ref{combined backwards direction estimate} and Fatou's lemma to obtain the first estimate.

    In order to prove the second estimate, we first notice that by the proof of Corollary \ref{combined backwards direction estimate}, we have that for any $l\geq0$ such that $2^{l+1}\leq\tau^{-1}:$
    \[\|\tau\partial_{\tau}\alpha_l\|^2_{H^s}+\|\alpha_l-\log(2^l\tau)\cdot\tau\partial_{\tau}\alpha_l\|^2_{H^s}\lesssim2^{-l}\cdot\bigg(\big\|\alpha_l\big\|^2_{H^{s+1}}+\big\|\partial_{\tau}\alpha_l\big\|^2_{H^s}+\big\|\chi_l\big\|^2_{H^s}\bigg)\bigg|_{1}\]
    In particular, this implies that for any $l\geq0$ with $2^{l+1}\leq\tau^{-1}:$
    \begin{equation}\label{alpha low freq growth estimate}
        \big\|\alpha_l\big\|^2_{H^s}\lesssim\big(1+|\log\tau|^2\big)\bigg(\big\|\alpha_l\big\|^2_{H^{s+1/2}}+\big\|\partial_{\tau}\alpha_l\big\|^2_{{H^{s-1/2}}}+\big\|\chi_l\big\|^2_{{H^{s-1/2}}}\bigg)\bigg|_{1}
    \end{equation}
    For any fixed $l\geq0$ we use this relation and Proposition \ref{high frequency backward estimate} to get:
    \[\big\|\chi_l(0)\big\|_{{H^{s+1/2}}}^2\lesssim\big\|\chi_l(1)\big\|^2_{{H^{s+1/2}}}+\int_0^{2^{-l-1}}\tau^2\big\|\alpha_l\big\|_{H^{s+1/2}}^2d\tau+\int^1_{2^{-l-1}}\tau^2\big\|\alpha_l\big\|_{H^{s+1/2}}^2d\tau\lesssim\]\[\lesssim\big\|\chi_l(1)\big\|^2_{{H^{s+1/2}}}+\bigg(\big\|\alpha_l\big\|^2_{H^{s+1}}+\big\|\partial_{\tau}\alpha_l\big\|^2_{{H^{s}}}+\big\|\chi_l\big\|^2_{{H^{s}}}\bigg)\bigg|_{\tau=1}\cdot\int_0^{2^{-l-1}}\tau^2\big(1+|\log\tau|^2\big)d\tau+\]\[+\bigg(\big\|\alpha_l\big\|^2_{H^{s+1/2}}+\big\|\partial_{\tau}\alpha_l\big\|^2_{H^{s-1/2}}+\big\|\chi_l\big\|^2_{H^{s-1/2}}\bigg)\bigg|_{1}\cdot\int^1_{2^{-l-1}}\tau d\tau\lesssim\bigg(\big\|\alpha_l\big\|^2_{H^{s+1/2}}+\big\|\partial_{\tau}\alpha_l\big\|^2_{{H^{s-1/2}}}+\big\|\chi_l\big\|_{{H^{s+1/2}}}^2\bigg)\bigg|_{1}\]
    This completes the proof of the second estimate.
\end{proof}

Finally, we also prove some estimates which are not sharp in terms of spatial regularity, but will nevertheless be useful in the next section:
\begin{proposition}\label{not sharp estimates}
    Let $\big(\alpha,\chi\big)$ be the smooth solution of the system (\ref{main equation})-(\ref{definition of chi}) with initial data at $\tau=1$ given by $\chi(1),\alpha(1),\partial_{\tau}\alpha(1)\in C^{\infty}(S^n).$ Then we have the estimates for $\tau\in(0,1]$:
    \[\big\|\alpha\big\|^2_{H^s}(\tau)\lesssim\big(1+|\log\tau|^2\big)\cdot\bigg(\big\|\alpha\big\|^2_{H^{s+1}}+\big\|\partial_{\tau}\alpha\big\|^2_{{H^{s}}}+\big\|\chi\big\|^2_{{H^{s}}}\bigg)\bigg|_{\tau=1}\]
    \[\big\|\tau\partial_{\tau}\alpha\big\|^2_{H^s}(\tau)\lesssim\bigg(\big\|\alpha\big\|^2_{H^{s+1}}+\big\|\partial_{\tau}\alpha\big\|^2_{{H^{s}}}+\big\|\chi\big\|^2_{{H^{s}}}\bigg)\bigg|_{\tau=1}\]
\end{proposition}
\begin{proof}
    We fix $\tau_0\in(0,1].$ Using the notation $\alpha_H=P_{\geq1/(4\tau_0)}\alpha$, we have the decomposition for some $l_0\geq1$:
    \[\alpha=\alpha_0+\ldots+\alpha_{l_0-1}+P_{(2^{l_0-1},1/4\tau_0)}\alpha+\alpha_H\]
    where $2^{-l_0-2}<\tau_0\leq2^{-l_0-1}.$ By the proof of Theorem \ref{main theorem backwards direction} and (\ref{alpha low freq growth estimate}), we have that for any $l\leq l_0:$
    \[\big\|\alpha_l\big\|^2_{H^s}(\tau_0)\lesssim\big(1+|\log\tau_0|^2\big)\bigg(\big\|\alpha_l\big\|^2_{H^{s+1/2}}+\big\|\partial_{\tau}\alpha_l\big\|^2_{{H^{s-1/2}}}+\big\|\chi_l\big\|^2_{{H^{s-1/2}}}\bigg)\bigg|_{\tau=1}\]
    \[\|\tau\partial_{\tau}\alpha_l\|^2_{H^s}(\tau_0)\lesssim\bigg(\big\|\alpha_l\big\|^2_{H^{s+1/2}}+\big\|\partial_{\tau}\alpha_l\big\|^2_{{H^{s-1/2}}}+\big\|\chi_l\big\|^2_{{H^{s-1/2}}}\bigg)\bigg|_{\tau=1}\]
    We note that $P_{l_0}P_{(2^{l_0-1},1/4\tau_0)}\alpha=P_{(2^{l_0-1},1/4\tau_0)}\alpha,$ so these estimates also apply to $P_{(2^{l_0-1},1/4\tau_0)}\alpha.$ To deal with the high frequencies, we recall that $\alpha_H$ satisfies the estimate in Proposition \ref{high frequency backward estimate}:
    \[\frac{1}{\tau}\big\|\alpha_H\big\|^2_{H^s}+\tau\big\|\partial_{\tau}\alpha_H\big\|^2_{H^s}\lesssim \bigg(\big\|\alpha_H\big\|^2_{H^{s+1}}+\big\|\partial_{\tau}\alpha_H\big\|^2_{H^s}+\big\|\chi_H\big\|^2_{H^s}\bigg)\bigg|_{\tau=1}\]
    Summing the low frequency estimates for all $l\leq l_0,$ and using the high frequency estimate, we obtain the conclusion.
\end{proof}

\subsection{Asymptotic Completeness}\label{asymptotic expansion subsection}

In this section, we prove that smooth solutions of the system (\ref{main equation})-(\ref{definition of chi}) defined on $(0,\infty)$ induce asymptotic data at $\tau=0$. Together with the estimates proved in Theorem \ref{main theorem backwards direction}, this completes the proof of the second statement in Theorem \ref{scattering theorem for main equation intro}, establishing \textit{asymptotic completeness}.

\begin{proposition}[Asymptotic completeness]\label{existence of asymptotic expansion smooth}
     Let $\big(\alpha,\chi\big)$ be the smooth solution of the system (\ref{main equation})-(\ref{definition of chi}) with initial data at $\tau=1$ given by $\chi(1),\alpha(1),\partial_{\tau}\alpha(1)\in C^{\infty}(S^n).$ Then there exist smooth $\chi(0),\mathcal{O},$ and $h$, which determine the asymptotic data of the solution at $\tau=0$.
\end{proposition}
\begin{proof}
    We write equation (\ref{main equation}) as:
    \[\partial_{\tau}^2\alpha+\frac{1}{{\tau}}\partial_{\tau}\alpha=-4q'\alpha+\frac{4}{({\tau}^2+1)^2}\cdot\Delta\alpha+f_1(\tau)\tau\partial_{\tau}\alpha+f_2(\tau)\tau^2\alpha+f_3(\tau)\cdot\bigg(\chi(1)-\int^1_{\tau}2\tau'\alpha(\tau')d\tau'\bigg)\]
    As before, we treat the RHS as an inhomogeneous term, so we can write the equation as:
    \begin{equation}\label{main equation alpha like on ODE and inhom term}
        \partial_{\tau}^2\alpha+\frac{1}{\tau}\partial_{\tau}\alpha=F\bigg(\alpha,\partial_{\tau}\alpha,\Delta\alpha,\int^1_{\tau}\tau'\alpha(\tau')d\tau',\chi(1)\bigg)
    \end{equation}
    We can solve this equation as an inhomogeneous ODE to obtain:
    \[\alpha(\tau)=\alpha(1)+\partial_{\tau}\alpha(1)\cdot\log\tau-\log\tau\cdot\int_{\tau}^1\tau'\cdot Fd\tau'+\int_{\tau}^1\tau'\log\tau'\cdot Fd\tau'\]
    We notice that using the estimates in Proposition \ref{not sharp estimates}, we get for any $s\geq1$:
    \[\bigg\|F\bigg(\alpha,\partial_{\tau}\alpha,\Delta\alpha,\int^1_{\tau}\tau'\alpha(\tau')d\tau',\chi(1)\bigg)\bigg\|_{H^s}\lesssim\big(1+|\log\tau|\big)\cdot\bigg(\big\|\alpha\big\|_{H^{s+3}}+\big\|\partial_{\tau}\alpha\big\|_{{H^{s+2}}}+\big\|\chi\big\|_{{H^{s+2}}}\bigg)\bigg|_{\tau=1}\]
    In particular, we can define the smooth functions:
    \[\chi(0)=\chi(1)-\int_0^12\tau'\alpha(\tau')d\tau',\  \mathcal{O}:=\partial_{\tau}\alpha(1)-\int_{0}^1\tau'\cdot Fd\tau',\ h:=\alpha(1)+\int_{0}^1\tau'\log\tau'\cdot Fd\tau'\]
    As a result, we have:
    \[\alpha(\tau)=\mathcal{O}\cdot\log\tau+h+\log\tau\cdot\int^{\tau}_0\tau'\cdot Fd\tau'-\int^{\tau}_0\tau'\log\tau'\cdot Fd\tau'\]
    Using our previous bound on $F,$ we see that $\alpha-\mathcal{O}\log({\tau})-h\in C^1_{\tau}([0,\infty))C^{\infty}(S^n)$ and:
    \[\alpha({\tau})-\mathcal{O}\log({\tau})-h=O\big({\tau}^2|\log({\tau})|^2\big),\ \partial_{\tau}\alpha({\tau})-\frac{\mathcal{O}}{\tau}=O\big({\tau}|\log({\tau})|^2\big)\text{ in }C^{\infty}(S^n).\]
    We conclude that $\big(\chi(0),\mathcal{O},h\big)$ determine the asymptotic data of $\big(\alpha,\chi\big)$ at $\tau=0.$
\end{proof}

In conclusion, the results proved in this section imply the existence and uniqueness of scattering states, and the asymptotic completeness statements in Theorem \ref{scattering theorem for main equation intro}. As a result, we can define the scattering map $\big(\chi(0),\mathcal{O},\mathfrak{h}\big)\mapsto\big(\chi(1),\alpha(1),\partial_{\tau}\alpha(1)\big)$. The estimates (\ref{forward estimate intro}) and (\ref{backward estimate intro}) imply that the scattering map extends by density as a Banach space isomorphism from $H^{s+\frac{1}{2}}(S^n)\times H^s(S^n)\times H^s(S^n)$ to $H^{s+\frac{1}{2}}(S^n)\times H^{s+\frac{1}{2}}(S^n)\times H^{s-\frac{1}{2}}(S^n)$ for any $s\geq1$. Thus, we complete the proof of Theorem \ref{scattering theorem for main equation intro}.

\section{Scattering for Self-similar Solutions of the Wave Equation in Minkowski Space}\label{minkowski space scattering section}

The main result of this section is the proof of Theorem \ref{scattering map in Minkowski theorem}, which establishes a scattering theory for smooth self-similar solutions of the wave equation in the $\{u<0,\ v>0\}$ region of Minkowski space $\mathbb{R}^{n+2}.$

The proof relies on the scattering result in Theorem \ref{scattering theorem for main equation intro} proved in Section \ref{main equation estimates section}, together with a series of compatibility relations which we briefly explain now. We recall that the system (\ref{main equation intro})-(\ref{definition of chi intro}) models the equation for $\partial_v^{\frac{n}{2}}\phi$ along $u=-1$ for $v\in [0,1]$, and similarly the equation for $\partial_u^{\frac{n}{2}}\phi$  along $v=1$ for $u\in [-1,0]$. Integrating these functions, we obtain a self-similar solution $\phi$ of the wave equation (\ref{wave equation phi}) up to the choice of $\frac{n}{2}$ integration constants. Determining the constants such that $\phi$ solves (\ref{wave equation phi}) at $\{v=0\}$ and $\{u=0\}$ implies a series of compatibility conditions.

In Section \ref{compatibility relations subsection} we derive the compatibility relations mentioned above. In Section \ref{minkowksi existence to ivp subsection} we prove the \textit{existence and uniqueness of scattering states}. In Section \ref{minkowksi asymptotic expansion subsection} we prove \textit{asymptotic completeness}. Finally, in Section \ref{minkowski scattering subsection} we complete the proof of Theorem \ref{scattering map in Minkowski theorem}, by constructing the \textit{scattering isomorphism}.

\subsection{Compatibility Relations}\label{compatibility relations subsection}
We assume that $\phi$ is a self-similar solution of the wave equation (\ref{wave equation phi}) in the region $\{u<0,\ v>0\}\subset\mathbb{R}^{n+2}.$ In this section we prove that under some regularity conditions, $\phi$ satisfies the compatibility relations referred to above. Based on this, we provide a description of the asymptotic behavior of the solution near $v=0$ and near $u=0.$

We assume for now that $\phi$ satisfies the smoothness condition:
\[\phi(-1,v)\in W^{\frac{n}{2},1}_{loc}\big([0,\infty);C^{\infty}(S^n)\big)\cap C^{\infty}((0,\infty)\times S^n)\]
By self-similarity, we have that $\phi$ solves (\ref{wave equation phi u=-1}), the wave equation along $u=-1.$ In particular, the computations in Section \ref{deriving model problem} imply that $\chi=\partial_v^{\frac{n}{2}-1}\phi,\ \alpha=\partial_v^{\frac{n}{2}}\phi$ are smooth solutions of the model problem from Section \ref{main equation estimates section}. By Proposition \ref{existence of asymptotic expansion smooth}, we know that there exist $\chi(0),\mathcal{O},h\in C^{\infty}(S)$ which determine the asymptotic data of $(\alpha,\chi)$ at $v=0.$ We note that we changed coordinates from $\tau$ to $v=\tau^2,$ and unless otherwise stated for the rest of the paper $(\alpha,\chi)$ will always be expressed in the $v$ coordinate.

We notice that the regularity assumptions on $\phi$ imply that along $u=-1$:
\begin{equation}\label{Taylor expansion for phi}
    \phi(-1,v)=\sum_{k=0}^{n/2-1}\frac{1}{k!}\phi_kv^k+\int_0^v\int_0^{v_{n/2}}\cdots\int_0^{v_2}\alpha(v_1)dv_1dv_2\ldots dv_{n/2}
\end{equation}
where we use the notation $\phi_k=\partial_v^k\phi(0).$ The above equation implies that for any $1\leq a\leq\frac{n}{2}-1:$
\begin{equation}\label{Taylor expansion for phi a}
    \partial_v^a\phi(-1,v)=\sum_{k=a}^{n/2-1}\frac{1}{(k-a)!}\phi_kv^{k-a}+\int_0^v\int_0^{v_{n/2-a}}\cdots\int_0^{v_2}\alpha(v_1)dv_1dv_2\ldots dv_{n/2-a}
\end{equation}

We use the above relations in order to evaluate equation (\ref{wave equation phi u=-1}) at $v=0,$ and we obtain the compatibility relation:
\[\bigg(1-\frac{n}{2}\bigg)\phi_1=\Delta\phi_0\]
Similarly, we evaluate equation (\ref{wave equation dv3 phi}) at $v=0,$ and we obtain the compatibility relation:
\[\bigg(2-\frac{n}{2}\bigg)\phi_2+2\phi_1=\Delta\phi_1\]
We repeat the same process in equation (\ref{wave equation dva phi}) for each $2\leq a\leq\frac{n}{2}-2,$  and we obtain the compatibility relations:
\[\bigg(a+1-\frac{n}{2}\bigg)\phi_{a+1}+q_a'\phi_a+q_a''\phi_{a-1}=\Delta\phi_a\]
We notice that to obtain the above relation in the case of $a=\frac{n}{2}-2,$ we use the asymptotic expansion of $\alpha=\partial_v^{\frac{n}{2}}\phi$, which implies $\lim_{v\rightarrow0}v\partial_v^{\frac{n}{2}}\phi(-1,v)=0.$

In summary, for each $1\leq a\leq\frac{n}{2}-1,$ we have the compatibility relation:
\begin{equation}\label{compatibility a}
    \phi_a=\Phi_a\big(\Delta^a\phi_0,\ldots,\Delta\phi_0\big),
\end{equation}
where each $\Phi_a$ is a first order multi-linear function such that the coefficient of $\Delta^a\phi_0$ is nonzero. In particular, we point out that $\chi(0)$ is determined by $\phi_0$ as well:
\begin{equation}\label{compatibility chi}
    \chi(0)=\phi_{\frac{n}{2}-1}=\Phi_{\frac{n}{2}-1}\big(\Delta^{\frac{n}{2}-1}\phi_0,\ldots,\Delta\phi_0\big)
\end{equation}

We consider now equation (\ref{wave equation dva phi}) for $a=\frac{n}{2}-1$:
\[v(v+1)^2\partial_v\alpha+\big[n(v+1)+p_av+q_a\big]v\alpha+(p'_av+q'_a)\chi+q_a''\partial_v^{a-1}\phi-\Delta\chi=0\]
Using the expansion (\ref{main expansion}), we obtain the compatibility relation:
\begin{equation}\label{compatibility O}
    \mathcal{O}=\Phi_{\frac{n}{2}}\big(\Delta^{\frac{n}{2}}\phi_0,\ldots,\Delta\phi_0\big)=2\Delta\chi(0)-2q_{\frac{n}{2}-1}'\chi(0)-2q_{\frac{n}{2}-1}''\phi_{\frac{n}{2}-2},
\end{equation}
where $\Phi_{\frac{n}{2}}$ is a first order multi-linear function such that the coefficient of $\Delta^{\frac{n}{2}}\phi_0$ is nonzero.

We prove the following result:
\begin{proposition}\label{asymptotic expansion phi proposition}
    For any $\phi$ self-similar solution of (\ref{wave equation phi}) in the region $\{u<0,\ v>0\},$ we have:
    \begin{itemize}
        \item If $\phi(-1,v)\in W^{\frac{n}{2},1}_{loc}\big([0,\infty);C^{\infty}(S^n)\big)\cap C^{\infty}((0,\infty)\times S^n)$, there exist unique $\phi_0,h\in C^{\infty}(S^n)$ such that for $u<0$:
        \begin{equation}\label{expansion for phi at v=0}
            \phi(u,v)=\sum_{k=0}^{n/2-1}\frac{1}{k!}\phi_k\bigg(\frac{v}{-u}\bigg)^k+\int_0^{\frac{v}{-u}}\int_0^{v_{n/2}}\cdots\int_0^{v_2}\alpha(v_1)dv_1dv_2\ldots dv_{n/2},
        \end{equation}
        where $\phi_1,\ldots,\phi_{\frac{n}{2}-1}\in C^{\infty}(S^n)$ are defined by the compatibility relation (\ref{compatibility a}), $\chi(0),\mathcal{O}\in C^{\infty}(S^n)$ are defined by the compatibility relations (\ref{compatibility chi})-(\ref{compatibility O}), and $\big(\alpha,\chi\big)$ is the smooth solution of (\ref{main equation})-(\ref{definition of chi}) with asymptotic initial data $\big(\chi(0),\mathcal{O},h\big).$
        \item If $\phi(u,1)\in W^{\frac{n}{2},1}_{loc}\big((-\infty,0];C^{\infty}(S^n)\big)\cap C^{\infty}((-\infty,0)\times S^n)$, there exist unique $\underline{\phi_0},\underline{h}\in C^{\infty}(S^n)$ such that for $v>0$:
        \begin{equation}\label{expansion for phi at u=0}
            \phi(u,v)=\sum_{k=0}^{n/2-1}\frac{1}{k!}\underline{\phi_k}\bigg(\frac{-u}{v}\bigg)^k+\int_0^{\frac{-u}{v}}\int_0^{u_{n/2}}\cdots\int_0^{u_2}\underline{\alpha}(u_1)du_1du_2\ldots du_{n/2},
        \end{equation}
        where $\underline{\phi_1},\ldots,\underline{\phi_{\frac{n}{2}-1}}\in C^{\infty}(S^n)$ are defined by the compatibility relation (\ref{compatibility a}), $\underline{\chi}(0),\underline{\mathcal{O}}\in C^{\infty}(S^n)$ are defined by the compatibility relations (\ref{compatibility chi})-(\ref{compatibility O}), and $\big(\underline{\alpha},\underline{\chi}\big)$ is the smooth solution of (\ref{main equation})-(\ref{definition of chi}) with asymptotic initial data $\big(\underline{\chi}(0),\underline{\mathcal{O}},\underline{h}\big).$
    \end{itemize}
\end{proposition}
\begin{proof}
    In the discussion preceding the proposition we already proved the first part of the statement. Equation (\ref{expansion for phi at v=0}) follows from (\ref{Taylor expansion for phi}) using the fact that:
    \[\phi(u,v)=\phi\bigg(-1,\frac{v}{-u}\bigg).\]
    
    We define $\psi(u,v)=\phi(-v,-u).$ Then $\psi(-1,v)\in W^{\frac{n}{2},1}_{loc}\big([0,\infty);C^{\infty}(S^n)\big)\cap C^{\infty}((0,\infty)\times S^n)$. By the first part, we know that there exist $\underline{\phi_0}:=\psi_0=\phi(0,1),\ \underline{h}\in C^{\infty}(S^n)$ which determine the asymptotic data of $\psi$ at $v=0.$ Let $\underline{\phi_1}:=\psi_1,\ldots,\underline{\phi_{\frac{n}{2}-1}}:=\psi_{\frac{n}{2}-1},\underline{\chi}(0),\underline{\mathcal{O}}\in C^{\infty}(S^n)$ satisfy the compatibility relations (\ref{compatibility a})-(\ref{compatibility O}). We denote by $\big(\underline{\alpha},\underline{\chi}\big)$ the smooth solution of (\ref{main equation})-(\ref{definition of chi}) with asymptotic initial data $\big(\underline{\chi}(0),\underline{\mathcal{O}},\underline{h}\big).$ In particular, we have $\underline{\alpha}=\partial_v^{\frac{n}{2}}\psi,\ \underline{\chi}=\partial_v^{\frac{n}{2}-1}\psi.$ Then we have that for $u<0:$
    \[\phi(-v,-u)=\psi(u,v)=\sum_{k=0}^{n/2-1}\frac{1}{k!}\underline{\phi_k}\bigg(\frac{v}{-u}\bigg)^k+\int_0^{\frac{v}{-u}}\int_0^{v_{n/2}}\cdots\int_0^{v_2}\underline{\alpha}(v_1)dv_1dv_2\ldots dv_{n/2}\]
\end{proof}

Based on this result, we can give a precise definition of asymptotic data at $\{v=0\}$ and $\{u=0\}$ for self-similar solutions of (\ref{wave equation phi}):
\begin{definition}
    Let $\phi\in C^{\infty}(S^n)$ be a smooth self-similar solution of (\ref{wave equation phi}) in the region $\{u<0,v>0\}.$ We say that $\big(\phi_0,h\big)$ determine the asymptotic data of $\phi$ at $v=0$ if the following conditions hold:
    \begin{itemize}
        \item $\phi(-1,v)\in W^{\frac{n}{2},1}_{loc}\big([0,\infty);C^{\infty}(S^n)\big)\cap C^{\infty}((0,\infty)\times S^n)$,
        \item for each $1\leq a\leq \frac{n}{2}-1,$ we have that $\partial_v^a\phi(-1,v)=\phi_a$ satisfies the compatibility relation (\ref{compatibility a}),
        \item the functions $\alpha=\partial_v^{\frac{n}{2}}\phi,\ \chi=\partial_v^{\frac{n}{2}-1}\phi$, which solve the system (\ref{main equation})-(\ref{definition of chi}), have asymptotic initial data $\big(\chi(0),\mathcal{O},h\big),$ where $\chi(0),\mathcal{O}$ satisfy the compatibility relations (\ref{compatibility chi})-(\ref{compatibility O}),
        \item $\phi$ satisfies (\ref{expansion for phi at v=0}).
    \end{itemize}
    We have an analogous definition for $\big(\underline{\phi_0},\underline{h}\big)$ to determine the asymptotic data of $\phi$ at $u=0,$ according to the second part of Proposition \ref{asymptotic expansion phi proposition}. Equivalently, we can define $\big(\underline{\phi_0},\underline{h}\big)$ to determine the asymptotic data of $\phi$ at $u=0,$ if $\big(\underline{\phi_0},\underline{h}\big)$ determine the asymptotic data of $\psi(u,v)=\phi(-v,-u)$ at $v=0.$
\end{definition}

\subsection{Existence and Uniqueness of Scattering States}\label{minkowksi existence to ivp subsection}

We prove the existence and uniqueness of solutions with given asymptotic initial data. We remark that our previous result regarding the asymptotic behavior of smooth solutions provides us a with guideline for constructing the solutions. The key ingredient is the existence and uniqueness of scattering states for the model problem (\ref{main equation})-(\ref{definition of chi}).

\begin{proposition}[Existence and uniqueness of scattering states]\label{existence of smooth phi}
    For any $\phi_0,h\in C^{\infty}(S^n)$, there exists a unique $\phi$ self-similar solution of (\ref{wave equation phi}), which satisfies the smoothness condition \[\phi(-1,v)\in W^{\frac{n}{2},1}_{loc}\big([0,\infty);C^{\infty}(S^n)\big)\cap C^{\infty}((0,\infty)\times S^n),\] such that $\big(\phi_0,h\big)$ determine the asymptotic data of $\phi$ at $v=0$. Moreover, for any $s\geq1$ we have:
    \begin{equation}\label{main estimate for phi forward direction}
        \bigg(\sum_{a=0}^{\frac{n}{2}}\big\|\partial_v^a\phi\big\|_{H^{s+1/2}}+\big\|\partial_{v}^{\frac{n}{2}+1}\phi\big\|_{{H^{s-1/2}}}\bigg)\bigg|_{(u,v)=(-1,1)}\lesssim \big\|\mathfrak{h}\big\|_{{H^{s}}}+\big\|\phi_0\big\|_{{H^{s+n}}}
    \end{equation}
    where $\mathfrak{h}:=h-(\log\nabla)\mathcal{O}:=h-\sum_{l=1}^{\infty}l\log2\cdot\mathcal{O}_l$ and $\mathcal{O}$ is defined in terms of $\phi_0$ by (\ref{compatibility O}).
\end{proposition}
\begin{proof}
    We consider the functions $\phi_1,\ldots,\phi_{\frac{n}{2}-1}\in C^{\infty}(S^n)$ defined by the compatibility relation (\ref{compatibility a}), and we consider $\chi(0),\mathcal{O}\in C^{\infty}(S^n)$ given by the compatibility relations (\ref{compatibility chi})-(\ref{compatibility O}). Let $\big(\alpha,\chi\big)$ be the smooth solution of (\ref{main equation})-(\ref{definition of chi}) with asymptotic initial data $\big(\chi(0),\mathcal{O},h\big).$ We define $\phi$ in the region $\{u<0,\ v\geq0\}$ according to (\ref{expansion for phi at v=0}). By this formula, we already know that $\phi$ is self-similar and $\phi(-1,v)\in W^{\frac{n}{2},1}_{loc}\big([0,\infty);C^{\infty}(S^n)\big)\cap C^{\infty}((0,\infty)\times S^n).$ Moreover, we compute using (\ref{definition of chi}):
    \[\partial_v^{\frac{n}{2}-1}\phi(-1,v)=\phi_{\frac{n}{2}-1}+\int_0^v\alpha(v_1)dv_1=\chi(0)+\int_0^v\alpha(v_1)dv_1=\chi(v)\]
    
    We check that $\phi$ is indeed a solution of (\ref{wave equation phi}). Since $\phi$ is self-similar, we know that the wave equation (\ref{wave equation phi}) in $\{u<0,\ v\geq0\}$ is equivalent to (\ref{wave equation phi u=-1}) on $\{u=-1,\ v\geq0\}.$ Using our computation in Section \ref{deriving model problem}, we have:
    \[\partial_v^{\frac{n}{2}}\bigg(v(v+1)^2\partial_v^2\phi+\bigg(1-\frac{n}{2}\bigg)(v+1)^2\partial_v\phi+nv(v+1)\partial_v\phi-\Delta\phi\bigg)=\]\[=v(v+1)^2\partial_v^2\alpha+(v+1)^2\partial_v\alpha+(pv+q)v\partial_v\alpha+(p'v+q')\alpha+q''\chi-\Delta\alpha=0\]
    Using the compatibility relations from the previous section we also have that for any $0\leq a\leq\frac{n}{2}-1:$
    \[\bigg[\partial_v^{a}\bigg(v(v+1)^2\partial_v^2\phi+\bigg(1-\frac{n}{2}\bigg)(v+1)^2\partial_v\phi+nv(v+1)^2\partial_v\phi-\Delta\phi\bigg)\bigg]\bigg|_{v=0}=0\]
    Moreover, we recall that:
    \[\partial_v^{\frac{n}{2}-1}\bigg(v(v+1)^2\partial_v^2\phi+\bigg(1-\frac{n}{2}\bigg)(v+1)^2\partial_v\phi+nv(v+1)\partial_v\phi-\Delta\phi\bigg)=\]\[=v(v+1)^2\partial_v\alpha+\big[n(v+1)+p_av+q_a\big]v\alpha+(p'_av+q'_a)\chi+q_a''\partial_v^{a-1}\phi-\Delta\chi\]
    The above is in $W^{1,1}_{loc}\big([0,\infty);C^{\infty}(S^n)\big)\cap C^{\infty}((0,\infty)\times S^n),$ since $\partial_v\big(v\partial_v\alpha\big)\in L^{1}_{loc}\big([0,\infty);C^{\infty}(S^n)\big)$ by (\ref{alpha}) and the asymptotic expansion of $\alpha$. Thus, we can apply the fundamental theorem of calculus repeatedly and obtain that $\phi$ defined by (\ref{expansion for phi at v=0}) solves (\ref{wave equation phi}).

    We notice that the definition of $\phi$ in the beginning of the proof implies that $\big(\phi_0,h\big)$ determine the asymptotic data of $\phi$ at $v=0$. To prove uniqueness we notice that given a solution with $\phi_0=h=0,$ then for all $1\leq a\leq\frac{n}{2}-1$ we have $\phi_a=0,$ and also $\alpha\equiv\chi\equiv0.$ Thus, by Proposition \ref{asymptotic expansion phi proposition} we obtain that the solution satisfies (\ref{expansion for phi at v=0}), so it vanishes identically.

    We notice that by Theorem \ref{improved regularity theorem} and the compatibility relations we have that:
    \[\bigg(\big\|\partial_v^{\frac{n}{2}}\phi\big\|_{H^{s+1/2}}+\big\|\partial_{v}^{\frac{n}{2}+1}\phi\big\|_{{H^{s-1/2}}}\bigg)\bigg|_{(u,v)=(-1,1)}\lesssim \big\|\mathfrak{h}\big\|_{{H^{s}}}+\big\|\mathcal{O}\big\|_{{H^{s}}}+\big\|\phi_{\frac{n}{2}-1}\big\|_{{H^{s}}}\lesssim \big\|\mathfrak{h}\big\|_{{H^{s}}}+\big\|\phi_0\big\|_{{H^{s+n}}}\]
    By equation (\ref{Taylor expansion for phi a}) we obtain that for any $0\leq a\leq \frac{n}{2}-1$ and any $l\geq0:$
    \[\big\|\partial_v^a\phi_l\big\|^2_{H^{s+1/2}}(-1,1)\lesssim \sum_{k=a}^{n/2-1}\big\|\big(\phi_k\big)_l\big\|^2_{H^{s+1/2}}+\bigg(\int_0^1\int_0^{v_{n/2-a}}\cdots\int_0^{v_2}\big\|\alpha_l(v_1)\big\|_{H^{s+1/2}}dv_1dv_2\ldots dv_{n/2-a}\bigg)^2\lesssim\]\[\lesssim\big\|\big(\phi_0\big)_l\big\|^2_{H^{s+n}}+\bigg(\int_0^1\big\|\alpha_l(v)\big\|_{H^{s+1/2}}dv\bigg)^2\lesssim\big\|\big(\phi_0\big)_l\big\|^2_{H^{s+n}}+\bigg(\int_0^1\tau\big\|\alpha_l(\tau)\big\|_{H^{s+1/2}}d\tau\bigg)^2\]
    where we denote $\tau^2=v$ as in Section \ref{main equation estimates section}, and we also use the compatibility relations. Finally, this implies:
    \[\big\|\partial_v^a\phi_l\big\|^2_{H^{s+1/2}}(-1,1)\lesssim\big\|\big(\phi_0\big)_l\big\|^2_{H^{s+n}}+\int_0^1\tau^2\big\|\alpha_l(\tau)\big\|^2_{H^{s+1/2}}d\tau\lesssim \big\|\mathfrak{h}_l\big\|_{{H^{s}}}^2+\big\|\big(\phi_0\big)_l\big\|_{{H^{s+n}}}^2\]
    where the last inequality follows by the proof of Theorem \ref{improved regularity theorem} and the compatibility relation for $\mathcal{O}$. Summing over all $l\geq0$ we obtain (\ref{main estimate for phi forward direction}).
\end{proof}

\subsection{Asymptotic Completeness}\label{minkowksi asymptotic expansion subsection}
In this section, we prove that smooth self-similar solutions of the wave equation induce asymptotic data at $v=0$, which establishes \textit{asymptotic completeness}. Moreover, we also prove estimates on the asymptotic data in terms of the Cauchy initial data.

\begin{proposition}[Asymptotic completeness]\label{existence of asymptotic expansion smooth phi}
    Let $\phi\in C^{\infty}(S^n)$ be a smooth self-similar solution of (\ref{wave equation phi}) in the region $\{u<0,v>0\}.$ Then $\phi(-1,v)\in W^{\frac{n}{2},1}_{loc}\big([0,\infty);C^{\infty}(S^n)\big)\cap C^{\infty}((0,\infty)\times S^n)$, and there exist $\phi_0,h\in C^{\infty}(S^n)$ that determine the asymptotic data of $\phi$ at $v=0.$ Moreover, for any $s\geq1$ we have the estimate:
    \begin{equation}\label{main estimate for phi backward direction}
    \big\|\mathfrak{h}\big\|_{{H^{s}}}+\big\|\phi_0\big\|_{{H^{s+n}}}\lesssim \bigg(\sum_{a=0}^{\frac{n}{2}}\big\|\partial_v^a\phi\big\|_{H^{s+1/2}}+\big\|\partial_{v}^{\frac{n}{2}+1}\phi\big\|_{{H^{s-1/2}}}\bigg)\bigg|_{(u,v)=(-1,1)}
    \end{equation}
\end{proposition}
\begin{proof}
    Along $\{u=-1\}$ we define $\alpha=\partial_v^{\frac{n}{2}}\phi,\chi=\partial_v^{\frac{n}{2}-1}\phi$. Then $(\alpha,\chi)$ is a smooth solution of (\ref{main equation})-(\ref{definition of chi}), so by Proposition \ref{existence of asymptotic expansion smooth} there exist $\chi(0),\mathcal{O},h\in C^{\infty}(S)$ which determine the asymptotic data of the solution at $v=0.$ Using the expansion of $\alpha$ at $v=0,$ we have that $\alpha\in L^1_{loc}\big([0,\infty);C^{\infty}(S)\big)\cap C^{\infty}((0,\infty)\times S),$ which implies the desired regularity statement for $\phi(-1,v).$ As a result, Proposition \ref{asymptotic expansion phi proposition} implies that $\big(\phi_0,h\big)$ determine the asymptotic expansion of $\phi$ at $v=0.$ We recall that by Theorem \ref{main theorem backwards direction} we have the estimate:
    \[\big\|\mathfrak{h}\big\|_{{H^{s}}}+\big\|\mathcal{O}\big\|_{{H^{s}}}+\big\|\chi(0)\big\|_{{H^{s+1/2}}}\lesssim \bigg(\big\|\alpha\big\|_{H^{s+1/2}}+\big\|\partial_{v}\alpha\big\|_{{H^{s-1/2}}}+\big\|\chi\big\|_{{H^{s+1/2}}}\bigg)\bigg|_{v=1}\]
    
    We prove by induction that for any $0\leq a\leq\frac{n}{2}-1$ we have:
    \begin{equation}\label{induction hyp}
        \big\|\phi_a\big\|_{{H^{s+1/2}}}\lesssim\sum_{k=a}^{n/2-1}\big\|\partial_v^k\phi\big\|_{{H^{s+1/2}}}\big|_{(-1,1)}+\big\|\alpha\big\|_{H^{s+1/2}}\big|_{v=1}+\big\|\partial_{v}\alpha\big\|_{{H^{s-1/2}}}\big|_{v=1}
    \end{equation}
    The case $a=\frac{n}{2}-1$ follows from Theorem \ref{main theorem backwards direction}. We assume now that (\ref{induction hyp}) holds for all $k\in[a+1,n/2-1]$. Using (\ref{Taylor expansion for phi a}), we have that for any $l\geq0:$
    \[\big\|(\phi_a)_l\big\|^2_{H^{s+1/2}}\lesssim\]\[\lesssim\big\|\partial_v^a\phi_l\big\|^2_{H^{s+1/2}}\big|_{(-1,1)}+\sum_{k=a+1}^{n/2-1}\big\|\big(\phi_k\big)_l\big\|^2_{H^{s+1/2}}+\bigg(\int_0^1\int_0^{v_{n/2-a}}\cdots\int_0^{v_2}\big\|\alpha_l(v_1)\big\|_{H^{s+1/2}}dv_1\ldots dv_{n/2-a}\bigg)^2\lesssim\]
    \[\lesssim\sum_{k=a}^{n/2-1}\big\|\partial_v^k\phi_l\big\|^2_{{H^{s+1/2}}}\big|_{(-1,1)}+\big\|\alpha_l\big\|^2_{H^{s+1/2}}\big|_{v=1}+\big\|\partial_{v}\alpha_l\big\|^2_{{H^{s-1/2}}}\big|_{v=1}+\bigg(\int_0^1\big\|\alpha_l(v)\big\|_{H^{s+1/2}}dv\bigg)^2\lesssim\]
    \[\lesssim\sum_{k=a}^{n/2-1}\big\|\partial_v^k\phi_l\big\|^2_{{H^{s+1/2}}}\big|_{(-1,1)}+\big\|\alpha_l\big\|_{H^{s+1/2}}^2\big|_{v=1}+\big\|\partial_{v}\alpha_l\big\|^2_{{H^{s-1/2}}}\big|_{v=1},\]
    where the last inequality follows by the proof of Theorem \ref{main theorem backwards direction} once we change variables in the integral to $\tau^2=v$. As a result, we proved (\ref{induction hyp}).

    Finally, we remark that the compatibility relation (\ref{compatibility O}) implies that:
    \[\big\|\phi_0\big\|_{{H^{s+n}}}\lesssim\big\|\phi_0\big\|_{H^{s+1/2}}+\big\|\mathcal{O}\big\|_{H^{s}}\]
    This is the essential elliptic estimate which allows us to recover the top order estimate for $\phi_0.$ Thus, Theorem \ref{main theorem backwards direction} and the estimate (\ref{induction hyp}) for $a=0$ imply (\ref{main estimate for phi backward direction}).
\end{proof}

\subsection{The Scattering Isomorphism}\label{minkowski scattering subsection}

In this section we complete the proof of the scattering theory for self-similar solutions of the wave equation in the $\{u<0,\ v>0\}$ region of Minkowski space $\mathbb{R}^{n+2},$ by constructing the scattering isomorphism between asymptotic data at $\{v=0\}$ and asymptotic data at $\{u=0\}$.

\begin{theorem}\label{scattering map in Minkowski theorem}
    For any $n\geq4$ even integer, we have a complete scattering theory for smooth self-similar solutions of the wave equation in the $\{u<0,\ v>0\}$ region of Minkowski space $\mathbb{R}^{n+2}$:

    \begin{enumerate}
    \item Existence and uniqueness of scattering states: for any $\phi_0,h\in C^{\infty}(S^n)$, there exists a unique smooth self-similar solution of (\ref{wave equation phi}) with asymptotic data at $v=0$ given by $\big(\phi_0,h\big)$, such that: \[\phi(-1,v)\in W^{\frac{n}{2},1}_{loc}\big([0,\infty);C^{\infty}(S^n)\big)\cap C^{\infty}((0,\infty)\times S^n),\]
    and the estimate (\ref{main estimate for phi forward direction}) holds.
    \item Asymptotic completeness: any smooth self-similar solution of (\ref{wave equation phi}) satisfies: \[\phi(u,1)\in W^{\frac{n}{2},1}_{loc}\big((-\infty,0];C^{\infty}(S^n)\big)\cap C^{\infty}((-\infty,0)\times S^n),\]
    induces asymptotic data at $u=0$ given by $\underline{\phi_0},\underline{h}\in C^{\infty}(S^n)$, and satisfies the estimate (\ref{main estimate for phi backward direction}).
    \item The scattering isomorphism: for the above solution, we define the scattering map $\big(\phi_0,\mathfrak{h}\big)\mapsto\big(\underline{\phi_0},\underline{\mathfrak{h}}\big)$, where $\mathfrak{h}:=h-(\log\nabla)\mathcal{O},\ \underline{\mathfrak{h}}:=\underline{h}-(\log\nabla)\underline{\mathcal{O}}$ and $\mathcal{O},\ \underline{\mathcal{O}}$ are defined in terms of $\phi_0,\underline{\phi_0}$ by (\ref{compatibility O}). Then, for any $s\geq1$ we have the estimate:
    \begin{equation}\label{scattering estimate Minkowski}        \big\|\underline{\mathfrak{h}}\big\|_{{H^{s}}}+\big\|\underline{\phi_0}\big\|_{{H^{s+n}}}\lesssim\big\|\mathfrak{h}\big\|_{{H^{s}}}+\big\|\phi_0\big\|_{{H^{s+n}}},
    \end{equation}
    and the scattering map extends as a Banach space isomorphism on $H^{s+n}(S^n)\times H^{s}(S^n)$ for any $s\geq1$.
\end{enumerate}
\end{theorem}
\begin{proof}
    The existence and uniqueness of scattering states follows by Proposition \ref{existence of smooth phi}. Moreover, for any $s\geq1$ we have the estimate (\ref{main estimate for phi forward direction}) holds.

    We define $\psi\in C^{\infty}$ in the region $\{u<0,\ v>0\}$ by $\psi(u,v)=\phi(-v,-u),$ so $\psi$ is also a smooth self-similar solution of (\ref{wave equation phi}). By Proposition \ref{existence of asymptotic expansion smooth phi}, we have that $\psi(-1,v)\in W^{\frac{n}{2},1}_{loc}\big([0,\infty);C^{\infty}(S^n)\big)\cap C^{\infty}((0,\infty)\times S^n)$, and there exist $\underline{\phi_0},\underline{h}\in C^{\infty}(S^n)$ that determine the asymptotic data of $\psi$ at $v=0.$ Equivalently, we have that $\phi(u,1)\in W^{\frac{n}{2},1}_{loc}\big((-\infty,0];C^{\infty}(S^n)\big)\cap C^{\infty}((-\infty,0)\times S^n)$, and $\big(\underline{\phi_0},\underline{h}\big)$ determine the asymptotic data of $\phi$ at $u=0$. Thus, we proved asymptotic completeness.

    According to Proposition \ref{existence of asymptotic expansion smooth phi}, we also have the the estimate:
    \begin{equation}
    \big\|\underline{\mathfrak{h}}\big\|_{{H^{s}}}+ \big\|\underline{\phi_0}\big\|_{{H^{s+n}}}\lesssim \bigg(\sum_{a=0}^{\frac{n}{2}}\big\|\partial_v^a\psi\big\|_{H^{s+1/2}}+\big\|\partial_{v}^{\frac{n}{2}+1}\psi\big\|_{{H^{s-1/2}}}\bigg)\bigg|_{(u,v)=(-1,1)}
    \end{equation}
    where $\underline{\mathfrak{h}}:=\underline{h}-(\log\nabla)\underline{\mathcal{O}}$ and $\underline{\mathcal{O}}$ is defined in terms of $\underline{\phi_0}$ by (\ref{compatibility O}). We use the self-similarity of $\phi$ to get:
    \[\bigg(\sum_{a=0}^{\frac{n}{2}}\big\|\partial_v^a\psi\big\|_{H^{s+1/2}}+\big\|\partial_{v}^{\frac{n}{2}+1}\psi\big\|_{{H^{s-1/2}}}\bigg)\bigg|_{(u,v)=(-1,1)}\lesssim\bigg(\sum_{a=0}^{\frac{n}{2}}\big\|\partial_u^a\phi\big\|_{H^{s+1/2}}+\big\|\partial_{u}^{\frac{n}{2}+1}\phi\big\|_{{H^{s-1/2}}}\bigg)\bigg|_{(u,v)=(-1,1)}\lesssim\]\[\lesssim\bigg(\sum_{a=0}^{\frac{n}{2}}\big\|\partial_v^a\phi\big\|_{H^{s+1/2}}+\big\|\partial_{v}^{\frac{n}{2}+1}\phi\big\|_{{H^{s-1/2}}}\bigg)\bigg|_{(u,v)=(-1,1)}\]
    which proves the estimate (\ref{main estimate for phi backward direction}).
    
    Combining inequalities (\ref{main estimate for phi forward direction}) and (\ref{main estimate for phi backward direction}), we obtain the estimate (\ref{scattering estimate Minkowski}). This holds for smooth asymptotic data, but by density we can extend the scattering map to asymptotic data in $H^{s+n}(S^n)\times H^{s}(S^n)$. Since the reverse inequality also holds, we conclude that the scattering map is an isomorphism.
\end{proof}

\section{Scattering for Solutions of the Wave Equation on de Sitter Space}\label{scattering in de sitter section}

In this section we prove Theorem \ref{main theorem intro}, by completing the construction of the scattering isomorphism between asymptotic data at $\mathcal{I}^-$ and asymptotic data at $\mathcal{I}^+$ for the wave equation on de Sitter space.

We first need to define a precise notion of asymptotic data for (\ref{wave equation de sitter intro}). This is based on the definitions from Section \ref{minkowski space scattering section} and the correspondence from Section \ref{set up section}. Let $\Tilde{\phi}$ be a solution of (\ref{wave equation de sitter intro}). We recall that according to Lemma \ref{correspondence lemma}, the function $\phi:\{u<0,v>0\}\subset\mathbb{R}^{n+2}\rightarrow\mathbb{R}$ defined by $\phi=\Tilde{\phi}\circ\pi$ is the corresponding self-similar solution of (\ref{wave equation phi}). In particular, we have:
\[\Tilde{\phi}\circ\frac{1}{2}\log(v)=\phi(-1,v),\ \Tilde{\phi}\circ\bigg(-\frac{1}{2}\log\bigg)(-u)=\phi(u,1)\]
Moreover, asymptotic data for $\Tilde{\phi}$ at $\mathcal{I}^-$ corresponds to asymptotic data for $\phi$ at $v=0,$ and similarly asymptotic data for $\Tilde{\phi}$ at $\mathcal{I}^+$ corresponds to asymptotic data for $\phi$ at $u=0.$ Based on this correspondence, we define:

\begin{definition}Let $\phi_0,\mathfrak{h}\in C^{\infty}(S^n)$. We say that $\Tilde{\phi}:\mathbb{R}\times S^n\rightarrow\mathbb{R}$ solves (\ref{wave equation de sitter intro}) with asymptotic data at $\mathcal{I}^-$ given by $\big(\phi_0,\mathfrak{h}\big)$ if $\Tilde{\phi}\circ\frac{1}{2}\log\in W^{\frac{n}{2},1}_{loc}\big([0,\infty);C^{\infty}(S^n)\big)\cap C^{\infty}((0,\infty)\times S^n)$ and:
    \begin{equation}\label{expansion for tilde phi at I-}
            \Tilde{\phi}(T)=\sum_{k=0}^{n/2-1}\frac{1}{k!}\phi_k\cdot e^{2kT}+\int_0^{e^{2T}}\int_0^{v_{n/2}}\cdots\int_0^{v_2}\alpha(v_1)dv_1dv_2\ldots dv_{n/2},
    \end{equation}
    where $\phi_1,\ldots,\phi_{\frac{n}{2}-1}\in C^{\infty}(S^n)$ are defined by the compatibility relation (\ref{compatibility a}), $\chi(0),\mathcal{O}\in C^{\infty}(S^n)$ are defined by the compatibility relations (\ref{compatibility chi})-(\ref{compatibility O}), and $\big(\alpha,\chi\big)$ is the smooth solution of (\ref{main equation})-(\ref{definition of chi}) with asymptotic initial data $\big(\chi(0),\mathcal{O},h\big),$ for $h:=\mathfrak{h}+(\log\nabla)\mathcal{O}.$ 
\end{definition}
Using the expansion of $\alpha$ at $v=0,$ we get that (\ref{expansion for tilde phi at I-}) implies the expansion of $\Tilde{\phi}$ at $\mathcal{I}^-$:
\[\Tilde{\phi}(T)=\sum_{k=0}^{n/2-1}\frac{1}{k!}\phi_k\cdot e^{2kT}+\frac{1}{(n/2)!}\mathcal{O}\cdot Te^{nT}+\frac{1}{(n/2)!}\Tilde{h}\cdot e^{nT}+O\big(T^2e^{(n+2)T}\big)\]
where $\Tilde{h}$ is obtained by renormalizing $h$ with a linear factor of $\mathcal{O}.$ This provides a rigorous proof of the formal expansion (\ref{expansion for phi intro}) in the introduction. 

We have a similar definition for asymptotic data at $\mathcal{I}^+:$

\begin{definition}Let $\underline{\phi_0},\underline{\mathfrak{h}}\in C^{\infty}(S^n)$. We say that $\Tilde{\phi}:\mathbb{R}\times S^n\rightarrow\mathbb{R}$ solves (\ref{wave equation de sitter intro}) with asymptotic data at $\mathcal{I}^+$ given by $\big(\underline{\phi_0},\underline{\mathfrak{h}}\big)$ if $\Tilde{\phi}\circ\big(-\frac{1}{2}\log\big)\in W^{\frac{n}{2},1}_{loc}\big([0,\infty);C^{\infty}(S^n)\big)\cap C^{\infty}((0,\infty)\times S^n)$ and:
    \begin{equation}\label{expansion for tilde phi at I+}
            \Tilde{\phi}(T)=\sum_{k=0}^{n/2-1}\frac{1}{k!}\underline{\phi_k}\cdot e^{-2kT}+\int_0^{e^{-2T}}\int_0^{u_{n/2}}\cdots\int_0^{u_2}\underline{\alpha}(u_1)du_1du_2\ldots du_{n/2},
    \end{equation}
    where $\underline{\phi_1},\ldots,\underline{\phi_{\frac{n}{2}-1}}\in C^{\infty}(S^n)$ are defined by the compatibility relation (\ref{compatibility a}), $\underline{\chi}(0),\underline{\mathcal{O}}\in C^{\infty}(S^n)$ are defined by the compatibility relations (\ref{compatibility chi})-(\ref{compatibility O}), and $\big(\underline{\alpha},\underline{\chi}\big)$ is the smooth solution of (\ref{main equation})-(\ref{definition of chi}) with asymptotic initial data $\big(\underline{\chi}(0),\underline{\mathcal{O}},\underline{h}\big),$ for $\underline{h}:=\underline{\mathfrak{h}}+(\log\nabla)\underline{\mathcal{O}}.$  
\end{definition}

Using the correspondence between $\Tilde{\phi}$ and $\phi$ that we explained above, the proof of Theorem \ref{main theorem intro} is a direct consequence of Theorem \ref{scattering map in Minkowski theorem}:

\textit{Proof of Theorem \ref{main theorem intro}.}
    We begin by proving the existence and uniqueness of scattering states. For any $\phi_0,\mathfrak{h}\in C^{\infty}(S^n)$, we set $h:=\mathfrak{h}+(\log\nabla)\mathcal{O},$ where $\mathcal{O}$ is defined in terms of $\phi_0$ by (\ref{compatibility O}). By Theorem \ref{scattering map in Minkowski theorem}, there exists a unique $\phi$ self-similar solution of (\ref{wave equation phi}) with asymptotic data at $v=0$ given by $\big(\phi_0,h\big)$. We define $\Tilde{\phi}:\mathbb{R}\times S^n\rightarrow\mathbb{R}$ to be $\Tilde{\phi}=\phi\circ\iota,$ so $\Tilde{\phi}$ solves (\ref{wave equation de sitter intro}) by Section \ref{change of coordinates subsection}. Using the self-similarity of $\phi$ we have:
    \[\phi(-1,v)=\phi\bigg(x=2\sqrt{v},T=\frac{1}{2}\log v\bigg)=\phi\bigg(x=1,T=\frac{1}{2}\log v\bigg)=\Tilde{\phi}\bigg(T=\frac{1}{2}\log v\bigg)\]
    As a result, all the requirements in the definition of a smooth solution of (\ref{wave equation de sitter intro}) with asymptotic data at $\mathcal{I}^-$ are satisfied as a consequence of Theorem \ref{scattering map in Minkowski theorem} and (\ref{expansion for phi at v=0}). In order to prove uniqueness, we consider a smooth solution of (\ref{wave equation de sitter intro}) with vanishing asymptotic data at $\mathcal{I}^-$. We associate $\phi=\Tilde{\phi}\circ\pi$, the corresponding self-similar solution of (\ref{wave equation phi}), which also has vanishing asymptotic data at $v=0$. Theorem \ref{scattering map in Minkowski theorem} implies that $\phi\equiv0,$ so we also obtain that $\Tilde{\phi}\equiv0.$

    Next, we prove asymptotic completeness. Let $\Tilde{\phi}:\mathbb{R}\times S^n\rightarrow\mathbb{R}$ be a smooth solution of (\ref{wave equation de sitter intro}), and consider $\phi=\Tilde{\phi}\circ\pi$ the corresponding self-similar solution of (\ref{wave equation phi}). Note that by the self-similarity of $\phi$ we have:
    \[\phi(u,1)=\phi\bigg(x=1,T=-\frac{1}{2}\log(-u)\bigg)=\Tilde{\phi}\bigg(T=-\frac{1}{2}\log(-u)\bigg)\]
    Moreover, we have from Theorem \ref{scattering map in Minkowski theorem} that $\phi(u,1)\in W^{\frac{n}{2},1}_{loc}\big((-\infty,0];C^{\infty}(S^n)\big)\cap C^{\infty}((-\infty,0)\times S^n)$ satisfies (\ref{expansion for phi at u=0}), and there exist $\underline{\phi_0},\underline{h}\in C^{\infty}(S^n)$ which determine the asymptotic data of $\phi$ at $u=0.$ Setting $\underline{\mathfrak{h}}:=\underline{h}-(\log\nabla)\underline{\mathcal{O}},$ where $\underline{\mathcal{O}}$ is defined in terms of $\underline{\phi_0}$ by (\ref{compatibility O}), we obtain that $\Tilde{\phi}$ solves (\ref{wave equation de sitter intro}) with asymptotic data at $\mathcal{I}^+$ given by $\big(\underline{\phi_0},\underline{\mathfrak{h}}\big)$.

    Finally, we define the scattering map $\big(\phi_0,\mathfrak{h}\big)\mapsto\big(\underline{\phi_0},\underline{\mathfrak{h}}\big)$, which is invertible since we can repeat the above arguments starting with data at $\mathcal{I}^+.$ The scattering map is bounded by  (\ref{scattering estimate Minkowski}), and similarly its inverse is also bounded. These estimates hold for smooth asymptotic data, but by density we can extend the scattering map to asymptotic data in $H^{s+n}(S^n)\times H^{s}(S^n)$. We conclude that the scattering map extends as an isomorphism on $H^{s+n}(S^n)\times H^s(S^n)$.
\qed

\section{Appendix}
\subsection{Littlewood-Paley decomposition}\label{Littlewood-Paley decomposition}

We present some of the elementary properties of the Laplacian operator on the sphere that we use throughout the paper for frequency decomposition arguments. We point out that all these properties also hold in the case of a general compact Riemannian manifold $\big(M^n,g_M\big)$. As remarked in the introduction, this implies that our arguments also apply in the case of the generalized de Sitter space.

In the paper we denote by $\Delta=\Delta_{g_{S^n}}$ the Laplace operator on the standard sphere. This satisfies the following properties, according to \cite{schoenyau}:
\begin{itemize}
    \item Each eigenvalue of $\Delta$ is real and has finite multiplicity.
    \item If we repeat each eigenvalue according to its multiplicity we have $\Sigma=\{\lambda_i\}_{i=0}^{\infty}$. Moreover, we have:
    \[0=\lambda_0<\lambda_1\leq\lambda_2\leq\cdots\]
    \item There exists an orthonormal basis of $L^2(S^n)$ given by $\{\varphi_i\}_{i=0}^{\infty},$ where $\varphi_i$ is an eigenfunction of $\lambda_i.$ Moreover, we have $\varphi_i\in C^{\infty}(S^n).$
\end{itemize}

We denote by $\langle\cdot,\cdot\rangle$ the standard inner product on $L^2(S^n)$. Using the orthonormal basis of $L^2(S^n)$, we can write for any $f\in L^2(S^n):$
\[\big\|f\big\|_{L^2}^2=\sum_{i=0}^{\infty}\big|\langle f,\varphi_i\rangle\big|^2\]
Similarly, for any $s\geq0$ we define the Sobolev space:
\[H^s(S^n)=\bigg\{f\in L^2(S^n):\ \big\|f\big\|_{H^s}^2:=\sum_{i=0}^{\infty}\big(1+\lambda_i^2\big)^{s}\cdot\big|\langle f,\varphi_i\rangle\big|^2<\infty\bigg\}\]

We define the frequency projection operators for any $f\in L^2(S^n):$
\[P_{\leq a}f=\sum_{\lambda_i\leq a}\langle f,\varphi_i\rangle\varphi_i,\ P_{\geq a}f=\sum_{\lambda_i\geq a}\langle f,\varphi_i\rangle\varphi_i,\ P_{(a,b]}f=\sum_{a<\lambda_i\leq b}\langle f,\varphi_i\rangle\varphi_i\]
We also introduce the notation $f_l=P_{(2^{l-1},2^{l}]}f$ for any $l\geq1,$ and $f_0=P_{\leq1}f.$ We obtain the Littlewood-Paley decomposition for any $f\in L^2(S^n):$
\[f=f_0+\sum_{l=1}^{\infty}f_l\]
We point out that the sequence $\{2^l\}$ can be replaced by something more general, without changing any of the arguments. Using this decomposition, we can equivalently write the norms on $H^s(S^n)$ as:
\[\big\|f\big\|_{H^s}^2\sim\sum_{l=0}^{\infty}2^{2ls}\big\|f_l\big\|_{L^2}^2\]

In the paper, we will construct Fourier multipliers at the level of $f_l,$ which is more robust than constructing Fourier multipliers at the level of each individual mode. For example, a key definition is:
\[\mathfrak{h}:=h-(\log\nabla)\mathcal{O}:=h-\sum_{l=1}^{\infty}l\log2\cdot\mathcal{O}_l\]
Alternatively, one could instead define $\mathfrak{h}':=h-\sum_{i=0}^{\infty}\log\langle\lambda_i\rangle\cdot\langle\mathcal{O},\varphi_i\rangle\varphi_i$. All our estimates would still hold in this case, since we have:
\[\mathfrak{h}-\mathfrak{h}'=\sum_{\lambda_i\in[0,1]}\log\langle\lambda_i\rangle\cdot\langle\mathcal{O},\varphi_i\rangle\varphi_i+\sum_{l=1}^{\infty}\sum_{\lambda_i\in(2^{l-1},2^l]}\log\frac{\langle\lambda_i\rangle}{2^l}\cdot\langle\mathcal{O},\varphi_i\rangle\varphi_i\]
which implies $\big\|\mathfrak{h}-\mathfrak{h}'\big\|_{H^s}^2\lesssim\big\|\mathcal{O}\big\|_{H^s}^2.$

\bibliographystyle{alpha}
\bibliography{refs}

\end{document}